\documentclass[12pt]{article}
\usepackage{amsfonts}
\usepackage{amsmath}
\allowdisplaybreaks[4]
\usepackage{amssymb}
\usepackage{amsthm}
\usepackage{mathrsfs}
\usepackage{srcltx}
\usepackage{accents}
\usepackage{amstext}
\usepackage{amsopn}
\usepackage{pdfpages}
\usepackage{caption}
\usepackage{enumitem}
\usepackage{color}
\usepackage{subfigure}
\usepackage{grffile}
\usepackage{float}
\usepackage{subfig}
\usepackage{overpic}
\usepackage{multirow}
\usepackage[pagewise]{lineno}
\usepackage[numbers]{natbib}
\graphicspath{{./figs/}}          
\DeclareGraphicsExtensions{.png,.jpg}  

\newtheorem{theorem}{Theorem}[section]
\newtheorem{lemma}[theorem]{Lemma}

\newtheorem{proposition}{Proposition} 

\theoremstyle{definition}

\numberwithin{equation}{section}

\begin{document}
	\date{}
	\title{ \bf\large{Effect of protection zone on the dynamics of  a diffusion-advection population-toxicant model}\footnote{This work was partially supported by grants from National Science Foundation of China (12371503, 12071382), Natural Science Foundation of
Chongqing (CSTB2022NSCQ-MSX0284, CSTB2024NSCQ-MSX0992).}}
	\author{Jing Gao,\ \ Xiaoli Wang,\footnote{Corresponding Author. Email: wxl711@swu.edu.cn }\ \ Guohong Zhang
		\\
		{\small School of Mathematics and Statistics, Southwest
			University, }}
	\maketitle
	\begin{abstract}
		{ This paper develops and analyzes a diffusion-advection model coupling population dynamics with toxicant transport, incorporating a boundary protection zone. For both upstream and downstream protection zone configurations, we investigate the combined influence of protected zones and key ecological factors on population persistence or extinction. Employing monotone dynamical system theory and eigenvalue analysis, we establish the global dynamics of the population-toxicant coexistence equilibrium. Furthermore, we characterize the parameter dependence governing the stability of the toxicant-only steady state, specifically examining the protected zone length, toxicant effect coefficient on population growth, per-unit contaminant discharge rate, toxicant input rate, diffusion/advection rates, and population natural growth rate. Finally, numerical simulations reveal the complex interplay between the protection zone and toxicant advection rate in significantly shaping population persistence domains.}

	\end{abstract}


\section {Introduction}

Aquatic ecosystems constitute vital natural environments essential for human survival. However, rapid industrialization and economic growth have continuously intensified environmental pressures. Large quantities of toxic substances, including untreated industrial wastewater, chemical detergents, and plastic fragments, are discharged directly into rivers and lakes, posing significant threats to aquatic organisms and severely disrupting ecosystem balance. As water pollution worsens globally, its remediation has become a critical issue of widespread concern.

Ecologists and mathematicians have employed mathematical modeling to investigate the long-term sustainability of aquatic communities and ecosystems \cite{hayashi2009,jager2010,spromberg2006}. Various models, including ordinary differential equation  models \cite{freedman1991,hallam1983,huang2015impact} and matrix population models \cite{erickson2014,hayashi2009,spromberg2005modeling},  have been developed.
However, these models often inadequately account for the significant influence of hydrological and physical properties of water bodies on population dynamics and toxicant behavior. In river ecosystems, once introduced, toxicants disperse via both diffusion and advective transport. Furthermore, given that the longitudinal scale of rivers typically far exceeds their cross-sectional dimensions, toxicants rapidly achieve cross-sectional homogeneity upon entering the flow. Reflecting this characteristic, advection-diffusion-reaction models in the horizontal dimension  are commonly adopted for theoretical studies
\cite{simon2015modeling,waghmare2017,zhang2017pollutant}. While such models for contaminated rivers effectively predict the spatiotemporal evolution of pollutants and their environmental impacts, they have paid less attention to the potential role of toxicants in actively shaping population dynamics.

 To specifically study the effects of environmental toxicants on populations in polluted rivers, Zhou and Huang \cite{huang2022} develop the following model  describing population-toxicant interactions in an advective environment:
\begin{equation}\label{eq0}
	\begin{cases}
		u_t=d_1u_{xx}-a_1u_x+u\left( r-cu-mw \right),&x\in \left( 0,L \right), t>0,
		\\
		w_t=d_2w_{xx}-a_2w_x+h(x)-qw-puw,&x\in \left(0,L \right), t>0,
		\\
		d_1u_x\left( 0,t \right) -a_1u\left( 0,t \right) =u_x\left( L,t \right) =0,&t>0,
		\\
		d_2w_x\left( 0,t \right) -a_2w\left( 0,t \right) =w_x\left( L,t \right) =0,&t>0,
		\\
		u\left( x,0 \right) =u_0\left( x \right) \geq,\not\equiv 0,
		\\
		w\left( x,0 \right) =w_0\left( x \right) \geq,\not\equiv 0.
	\end{cases}
	\end{equation}
Here $u(x, t)$ and $w(x, t)$ represent the population density
and the toxicant concentration at location $x$ and time $t$, respectively, and the function $h(x)$ is the rate of the exogenous
input of the toxicant into the river. In \eqref{eq0}, $u(x, t)$ and $w(x, t)$ satisfy a Robin condition at the upstream boundary $x=0$ and a Neumann condition at the downstream boundary $x= L$. For the interpretations of $d_1, d_2, r, c, m, p, q, L >0$ and $a_1, a_2\geq0$, one can refer to \cite{huang2022}. Based on eigenvalue analysis, \cite{huang2022} derived sufficient conditions for population persistence or extinction and numerically investigated the effects of factors like toxicant input, flow velocity, and diffusion/advection characteristics on population persistence and spatial distribution.
With a change of variables and assume that $h(x)\equiv h>0$, model \eqref{eq0} can be rewritten as:
 \begin{equation}\label{eq00}
	\begin{cases}
		u_t=d_1u_{xx}-a_1u_x+u\left( r-u-mw \right),&x\in \left( 0,L \right), t>0,
		\\
		w_t=d_2w_{xx}-a_2w_x+h-qw-puw,&x\in \left(0,L \right), t>0,
		\\
		d_1u_x\left( 0,t \right) -a_1u\left( 0,t \right) =u_x\left( L,t \right) =0,&t>0,
		\\
		d_2w_x\left( 0,t \right) -a_2w\left( 0,t \right) =w_x\left( L,t \right) =0,&t>0,
		\\
		u\left( x,0 \right) =u_0\left( x \right) \geq,\not\equiv 0,
		\\
		w\left( x,0 \right) =w_0\left( x \right) \geq,\not\equiv 0.
	\end{cases}
	\end{equation}
 Wang \cite{wangqi2023} extensively analyzed the dependence of the toxicant-only semi-trivial steady state of model \eqref{eq00} on multiple parameters by classifying cases based on the presence of advection in the species, toxicant, or both.
 Xing et al. \cite{xing2024}  extended  model \eqref{eq0} by incorporating the assumption that populations tend to move away from high toxicant concentrations, introducing a toxicant-taxis term. Their study investigated the existence of positive steady states and the influence of chemotaxis and population advection velocity on the stability of the pure-toxicant steady state, offering novel insights into toxicant-population dynamics.



To enhance the survival prospects of aquatic organisms in polluted rivers or streams, establishing protected zones represents a key strategy. This involves engineered pollutant remediation and targeted detoxification within a designated river section. The objectives are twofold: to ensure that species inhabiting this zone can grow unimpeded by toxicant, and to maintain unrestricted diffusion of aquatic organisms within the protected area.
  Critical questions then arise: How does the implementation of such protected zones influence population persistence or extinction dynamics within contaminated riverine ecosystems? Furthermore, how do ecological factors modulate population responses under protected area regimes?
   To address these questions, inspired by related works on protection zones in predator-prey and competition models \cite{cui2014strong,du2008diffusive,du2006diffusive,he2017protection,oeda2011effect,oeda2017steady,duiliu2024,wang2013effect,zaw2022dynamics}, we extend the framework of model \eqref{eq00} to investigate a population-toxicant model incorporating a protected zone. Empirical evidence suggests that establishing a single large protected area is generally more economical and practical than creating multiple smaller zones of equivalent total size
   \cite{du2006diffusive}. Consequently, this study focuses on a population-toxicant model featuring a single protected zone. Specifically, we examine the ecological implications of situating this protected zone at two distinct locations: upstream and downstream.

A spatiotemporal model to assess the impact of toxicants on aquatic populations in a polluted river with \textit{a  upstream protected zone} under advective conditions is
\begin{equation}\label{eq1.1}
	\begin{cases}
		u_t=d_1u_{xx}-a_1u_x+u\left( r-u-m\chi _{\left( x_1,L \right]}w \right),&x\in \left( 0,L \right) ,t>0,
		\\
		w_t=d_2w_{xx}-a_2w_x+h-qw-puw,&x\in \left( x_1,L \right) ,t>0,
		\\
		d_1u_x\left( 0,t \right) -a_1u\left( 0,t \right) =u_x\left( L,t \right) =0,&t>0,
		\\
		d_2w_x\left( x_1,t \right) -a_2w\left( x_1,t \right) =w_x\left( L,t \right) =0,&t>0,
		\\
		u\left( x,0 \right) =u_0\left( x \right) \geq,\not\equiv 0,&0<x<L,
		\\
		w\left( x,0 \right) =w_0\left( x \right) \geq,\not\equiv 0,&x_1<x<L.

	\end{cases}
	\end{equation}
	Here, $\chi _{\left( x_1,L \right]}$ is defined as
	\begin{equation*}
		\chi _{\left( x_1,L \right]}=
		\begin{cases}
			0,&0 \le x\le x_1,\\
			1,&x_1<x\le L.
		\end{cases}
	\end{equation*}
  This means that $(0, x_1)$, imposed  a Robin condition  on the toxicant at  $x_1$, represents a  upstream protected zone where aquatic species are unaffected by toxicant.

  Alternatively, a spatiotemporal model of populations exposed to toxicants in a polluted river featuring \textit{a downstream protected zone} within an advective environment is
  \begin{equation}\label{eq1.2}
  	\begin{cases}
  		u_t=d_1u_{xx}-a_1u_x+u\left( r-u-m\chi _{\left[ 0,x_2 \right)}w \right),&x\in \left( 0,L \right) ,t>0,
  		\\
  		w_t=d_2w_{xx}-a_2w_x+h-qw-puw,&x\in \left( 0,x_2 \right) ,t>0,
  		\\
  		d_1u_x\left( 0,t \right) -a_1u\left( 0,t \right) =u_x\left( L,t \right) =0,&t>0,
  		\\
  		d_2w_x\left( 0,t \right) -a_2w\left( 0,t \right) =d_2w_x\left( x_2,t \right) -a_2w\left( x_2,t \right) =0 ,&t>0,
  		\\
  		u\left( x,0 \right) =u_0\left( x \right) \geq,\not\equiv 0,&0<x<L,
  		\\
  		w\left( x,0 \right) =w_0\left( x \right) \geq,\not\equiv 0,&0<x<x_2.
  	\end{cases}
  \end{equation}
  Here, $\chi _{\left[ 0,x_2 \right)}$ is defined as
  \begin{equation*}
  	\chi _{\left[ 0,x_2 \right)}=
  	\begin{cases}
  		1,&0 \le x < x_2,\\
  		0,&x_2\le x\le L.
  	\end{cases}
  \end{equation*}
  Then $(x_2, L)$, imposed a Robin condition  on the toxicant at $x_2$, represents a downstream protected zone where aquatic species are unaffected by toxicant.

This paper is organized as follows. Section 2 presents the main results for models \eqref{eq1.1} and \eqref{eq1.2}. The proofs of these results are provided in Section 3. Section 4 numerically analyzes the combined effects of toxicant advection and protected zone size interactions on population persistence, as well as the distributions of the population and toxicant. Finally, Section 5 summarizes our conclusions and discusses future work.


\section{Main results}

Systems \eqref{eq1.1} and \eqref{eq1.2} admit a unique semi-trivial steady state, denoted by $(0,\tilde{w})$ and $(0,\bar{w})$, respectively, where $\tilde{w}$ and $\bar{w}$ will be defined  in Lemma \ref{lemma2.1}. Similarly to \cite{wangqi2023,huang2022}, the following proposition on the global dynamics of systems~\eqref{eq1.1} and \eqref{eq1.2} is true.
\begin{proposition}\label{propA}
\begin{description}
  \item[$(i)$] If system~\eqref{eq1.1} (resp. \eqref{eq1.2}) has no positive steady states, then the semi-trivial steady state $(0,\tilde{w})$ (resp. $(0,\bar{w})$) is globally asymptotically stable among all nonnegative and nontrivial initial conditions.
  \item[$(ii)$] If the semi-trivial steady state  $(0,\tilde{w})$ (resp. $(0,\bar{w})$)  is linearly unstable, then system~\eqref{eq1.1} (resp.\eqref{eq1.2}) has at least one local stable coexistence steady state; Furthermore, if every coexistence
	steady state is linearly stable, then there is a unique coexistence steady state which is globally asymptotically stable among all nonnegative and nontrivial initial conditions.
\end{description}
\end{proposition}

Define
 \begin{equation*}
 \begin{aligned}
 &L^*=2d_1\frac{\pi-\arctan\frac{\sqrt{4d_1r-a_1^2}}{a_1}}{\sqrt{4d_1r-a_1^2}},\\
 	&L_1^*=
 	\begin{cases}
 	2d_1\frac{\arctan\frac{a_1\sqrt{4d_1r-a_1^2}}{2d_1r-a_1^2}}{\sqrt{4d_1r-a_1^2}},\quad&0<a_1\le \sqrt{2d_1r},\\
 	2d_1\frac{\pi+\arctan\frac{a_1\sqrt{4d_1r-a_1^2}}{2d_1r-a_1^2}}{\sqrt{4d_1r-a_1^2}},\quad&\sqrt{2d_1r}<a_1< \sqrt{4d_1r}.
 	\end{cases}
 \end{aligned}
 \end{equation*}
Then $L^*>L_1^*$ \cite{lou2015evolution}.

 First, we give our main results on the dynamics of model \eqref{eq1.1}.

 \begin{theorem}\label{th1.1}
 	If $a_1\ge\sqrt{4d_1r}$ or $0< a_1<\sqrt{4d_1r}$ and $L<L^*_1$, then $(0,\tilde{w})$ is globally asymptotically stable. If $0< a_1<\sqrt{4d_1r}$ and $L>L_1^*$, then there exist critical values $L^*, m^*(x_1), q^*(x_1), h^*(x_1)$ such that the following statements hold:
 \begin{description}
   \item[$(i)$] If $x_1\ge L^*$ $(\text{large protection zone})$, then $(0,\tilde{w})$ is unstable for any $m,q,h>0$;
   \item[$(ii)$] If $x_1< L^*$ $(\text{small protection zone})$, then the following conclusions are true:
   \begin{description}
     \item[$(ii.1)$] If $m>m^*(x_1)$, then $(0,\tilde{w})$ is linearly stable; if $0<m<m^*(x_1)$, then $(0,\tilde{w})$ is unstable, where $m^*(x_1)$ is strictly increasing with respect to $x_1\in [0, L^*)$;
     \item[$(ii.2)$] If $0<q<q^*(x_1)$, then $(0,\tilde{w})$ is linearly stable; if $q>q^*(x_1)$, then $(0,\tilde{w})$ is unstable, where $q^*(x_1)$ is strictly decreasing with respect to $x_1\in [0, L^*)$;
     \item[$(ii.3)$] If  $h>h^*(x_1)$, then $(0,\tilde{w})$ is linearly stable; if $0<h<h^*(x_1)$, then $(0,\tilde{w})$ is unstable, where $h^*(x_1)$ is strictly increasing with respect to $x_1\in [0, L^*)$.
   \end{description}
 \end{description}
 \end{theorem}

\begin{figure}[h]
	\centering
	\subfigure[]{\includegraphics[width=5cm,height=5cm]{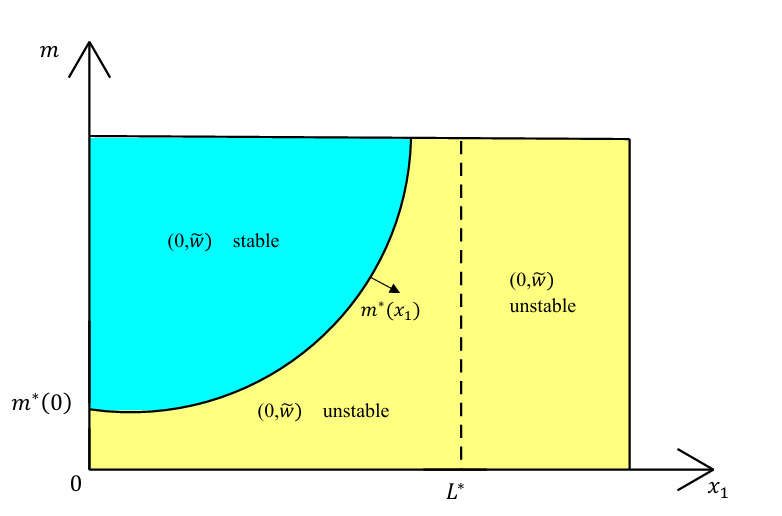}} \hfill
	\subfigure[]{\includegraphics[width=5cm,height=5cm]{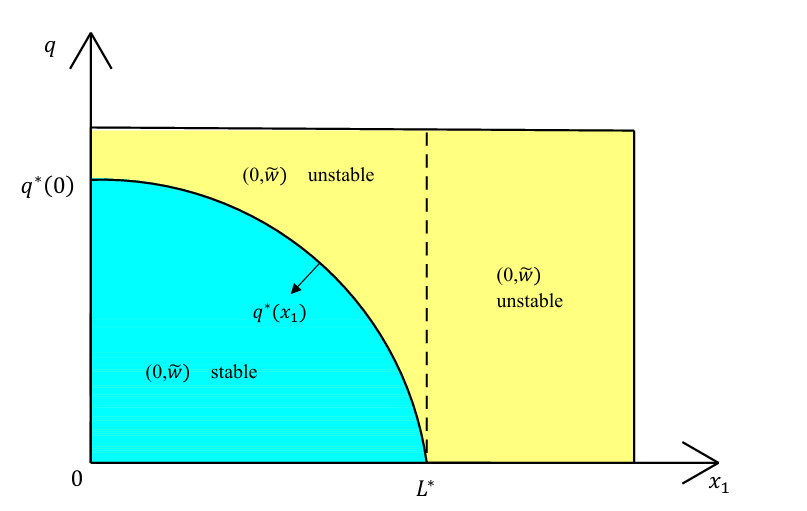}} \hfill
	\subfigure[]{\includegraphics[width=5cm,height=5cm]{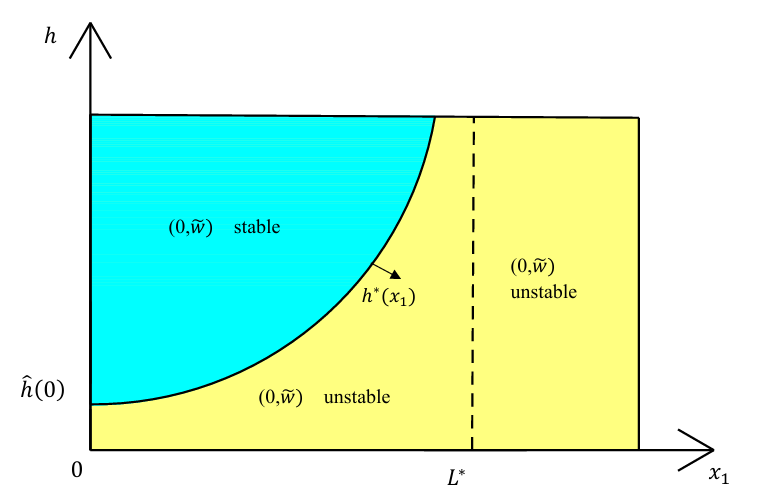}} \\
	\caption{Effects of the upstream protection zone length ($x_1$), the toxicant effect coefficient on population growth ($m$), the contaminant discharge rate per unit ($q$), and the toxicant input rate ($h$) on the stability of $(0,\tilde{w})$ when $0< a_1<\sqrt{4d_1r}$ and $L>L_1^*$. The yellow region indicates the instability of $(0,\tilde{w})$, implying the persistence of the population, while the blue region represents the stability of $(0,\tilde{w})$, indicating the extinction of the population when rare.}
	\label{fig:1}
\end{figure}

Theorem \ref{th1.1} shows that the population face extinction due to its large advection rate ($a_1\ge\sqrt{4d_1r}$) or small  advection rate while a short habitat length ($0< a_1<\sqrt{4d_1r}$, $L<L^*_1$).  When the advection rate of the population is small ($0< a_1<\sqrt{4d_1r}$) and the  river
 length is large $(L>L_1^*)$, Theorem \ref{th1.1} characterizes the influence of upstream protection zone length on the stability of the equilibrium $(0,\tilde{w})$ (see Figure \ref{fig:1}). Specifically,
\begin{itemize}
  \item When the protected area exceeds a critical spatial threshold $(x_1\ge L^*)$, population persistence occurs irrespective of toxicant exposure levels.
  \item Below this threshold $(x_1<L^*)$, persistence becomes dependent on additional factors. Specifically, if any of the following occurs: the toxicant effect coefficient on population growth ($m$) is sufficiently high, the contaminant discharge rate per unit ($q$) is sufficiently low, or the toxicant input rate ($h$) is sufficiently high, then a larger critical protection zone size is required to ensure persistence.
\end{itemize}

  \begin{theorem} \label{th1.2}
  	When $0< a_1<\sqrt{4d_1r}$ and $L>L_1^*$, for a small protection zone $x_1< L^*$, the following statements are true:
  \begin{description}
    \item[$(i)$] If all parameters except for $a_2$ and $h$ are fixed, and  $(a_2,h)\in \{a_2\ge 0,h<h^{**}\}\cup \:\{a_2>a_2^*(h), h>h^{**}\}$, then $(0,\tilde{w})$ is unstable; while for $(a_2,h)\in\{0\leq a_2<a_2^*(h), h>h^{**}\}$,  $(0,\tilde{w})$ is linearly stable, where $a_2^*(h)$ is strictly increasing with respect to $h\in [h^{**}, +\infty)$;
    \item[$(ii)$] If all parameters except for $a_2$ and $q$ are fixed, and  $(a_2,q)\in \{a_2\ge 0,q>q^{**}\}\cup \:\{a_2>\bar{a}_2(q),0<q<q^{**}\}$, then $(0,\tilde{w})$ is unstable; while for $(a_2,q)\in \{0\leq a_2<\bar{a}_2(q)\text{,}0<q<q^{**}\}$,  $(0,\tilde{w})$ is linearly stable, where $\bar{a}_2(q)$ is strictly decreasing with respect to $q\in [0,q^{**}]$;
    \item[$(iii)$]  If all parameters except for $a_2$ and $m$ are fixed, and  $(a_2,m)\in \{a_2\ge 0,m<m^{**}\}\cup \:\{a_2>\hat{a}_2(m), m>m^{**}\}$, then $(0,\tilde{w})$ is unstable; while for $(a_2,m)\in \{0\leq a_2<\hat{a}_2(m)\text{,}m>m^{**}\}$, $(0,\tilde{w})$ is linearly stable, where $\hat{a}_2(m)$ is strictly increasing with respect to $m\in [m^{**}, +\infty)$;
    \item[$(iv)$] If all parameters except for $h$ and $m$ are fixed, and  $(h,m)\in\{0<h<\hat{h}(m), m>0\}$, then $(0,\tilde{w})$ is unstable; while for $(h,m)\in\{h>\hat{h}(m)\text{,} m>0\}$,  $(0,\tilde{w})$ is linearly stable, where $\hat{h}(m)$ is strictly decreasing with respect to $m>0$;
    \item[$(v)$] If all parameters except for $q$ and $m$ are fixed, and the pair $(q,m)\in\{q>\hat{q}(m)\text{,} m>0\}$, then $(0,\tilde{w})$ is unstable; while for $(q,m)\in\{0<q<\hat{q}(m), m>0\}$, $(0,\tilde{w})$ is linearly stable,  where $\hat{q}(m)$ is strictly increasing with respect to $m>0$.
  \end{description}
  \end{theorem}


\begin{figure}[h]
  \centering
\begin{minipage}[c]{0.33\textwidth}
\centering
 \subfigure[]{
  \includegraphics[width=5cm,height=5cm]{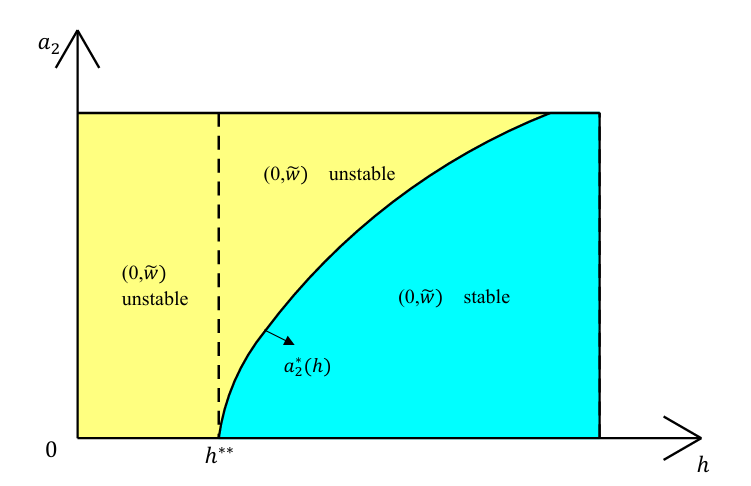}
  }
\end{minipage}%
\begin{minipage}[c]{0.33\textwidth}
\centering
 \subfigure[]{
  \includegraphics[width=5cm,height=5cm]{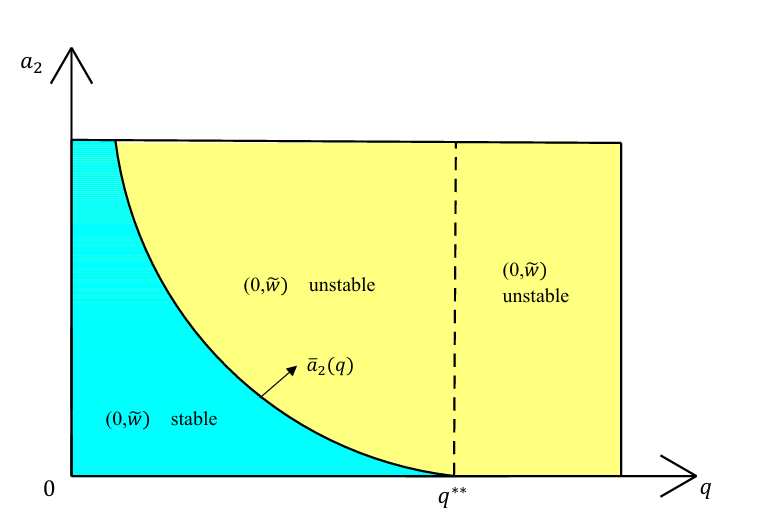}
 }
\end{minipage}%
\begin{minipage}[c]{0.33\textwidth}
\centering
 \subfigure[]{
  \includegraphics[width=5cm,height=5cm]{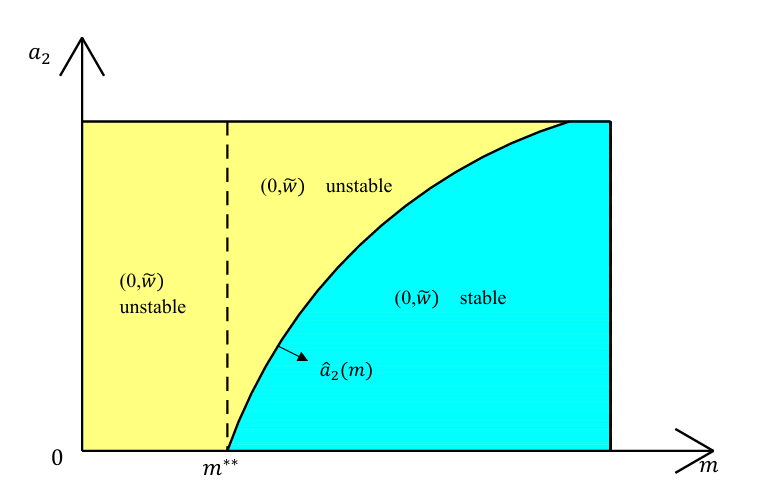}
 }
\end{minipage}
\begin{minipage}[c]{0.33\textwidth}
\centering
 \subfigure[]{
  \includegraphics[width=5cm,height=5cm]{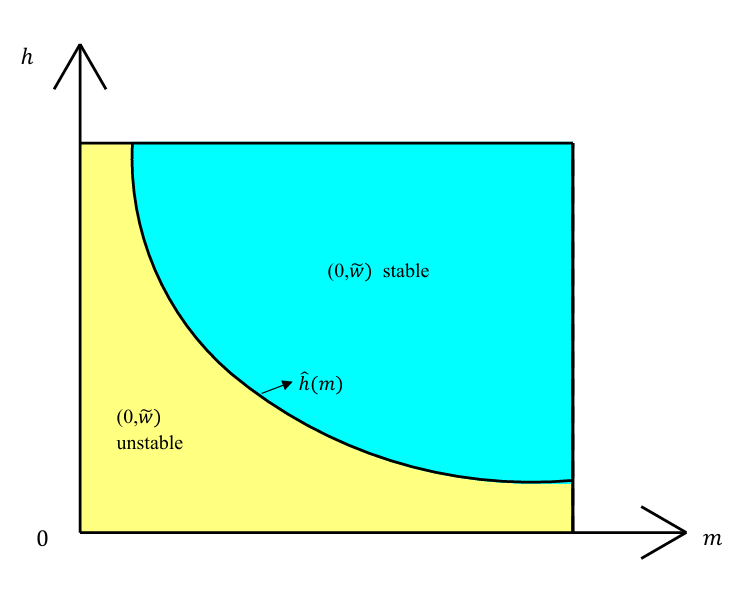}
 }
\end{minipage}%
\begin{minipage}[c]{0.33\textwidth}
\centering
 \subfigure[]{
  \includegraphics[width=5cm,height=5cm]{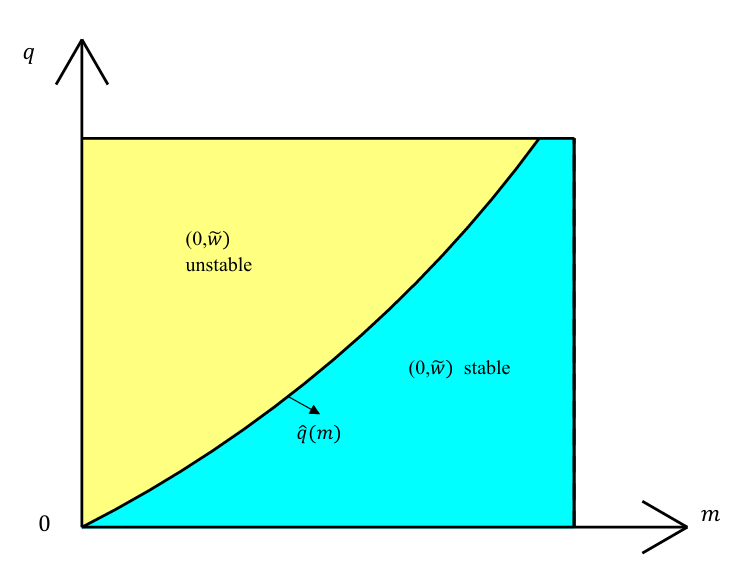}
 }
\end{minipage}
\caption{The stability distribution of $(0,\tilde{w})$ in the parameter spaces $h-a_2$, $q-a_2$, $m-a_2$, $m-h$, $m-q$ respectively under a small protection zone $(x_1<L^*)$ when $0< a_1<\sqrt{4d_1r}$ and $L>L_1^*$. $(0,\tilde{w})$ is unstable in the yellow region and stable in the blue region.}\label{fig:2}
\end{figure}

  $\quad$Theorem~\ref{th1.2} describes the stability distribution of $(0,\tilde{w})$ in the parameter spaces $h-a_2$, $q-a_2$, $m-a_2$, $m-h$, $m-q$ respectively under a small protection zone $(x_1<L^*)$ when the advection rate of the population is small ($0< a_1<\sqrt{4d_1r}$) and the  river
 length is large $(L>L_1^*)$ (see Figure \ref{fig:2}). Biologically,
  	\begin{itemize}
  	\item If the toxicant input rate $h$ or the effect of toxicant on population growth $m$ is small $(h<h^{**}$ or $ m<m^{**})$, then population $u$ achieves persistent survival for any advection rate of the toxicant  $(a_2>0)$. If the toxicant input or the effect of toxicant on population growth is large $(h>h^{**}$ or $m>m^{**})$, then the survival   of population $u$ becomes habitat-dependent, that is, the population $u$ can persist for large $a_2 (a_2>a_2^*(h)$, resp. $a_2>\hat{a}_2(m))$, while extinct when rare for small $a_2 (0\leq a_2<a_2^*(h)$, resp.  $0\leq a_2<\hat{a}_2(m))$.
  	\item If the contaminant export rate $q$ is large $(q>q^{**})$, then the population $u$ achieves persistent survival for any advection rate of the toxicant  $(a_2>0)$. If  $q$ is small $(q<q^{**})$, then the persistence of population $u$ becomes habitat-dependent, that is, the population $u$ can persist for large $a_2 (a_2>\bar{a}_2(q) )$ , while extinct when rare for small $a_2 (0\leq a_2<\bar{a}_2(q) )$.
  	\item  For sufficiently small  $m, h>0$, the population $u$ can always persist. The critical contamination threshold $\hat{h}(m)$ for the persistence of population $u$ decreases monotonically with respect to $m$.
  	\item  For fixed $m>0$,  the population $u$ can always persist when $q$ is sufficiently large. The critical detoxification threshold $\hat{q}(m)$ increases monotonically with respect to $m$.
  \end{itemize}

  \begin{theorem}\label{th1.3}
  Assume that all parameters except for $a_1$, $d_1$ and $r$ are fixed. Let $r^*=\frac{m\int_{x_1}^{L}\tilde{w}dx}{L}$. Then we have the following results:
    \begin{description}
      \item[$(i)$] If $r\ge r^*$, given $d_1$, then there exists $a_1^*>0$ satisfying $\lambda _1\left( d_1,a_1^*,r-m\chi _{\left( x_1,L \right]}\tilde{w},\left( 0,L \right) \right) =0$ such that $(0,\tilde{w})$ is linearly stable when $a_1>a_1^*$ and unstable when $0<a_1<a_1^*$;
      \item[$(ii)$] If $r<r^*$, then there exists $d_1^*(r)>0$ satisfying $\lambda _1\left( d_1^*(r),0,r-m\chi _{\left( x_1,L \right]}\tilde{w},\left( 0,L \right) \right) =0$ such that the following statements are true:
          \begin{description}
            \item[$(ii.1)$] If $d_1\ge d_1^*(r)$, then $(0,\tilde{w})$ is linearly stable for any $a_1>0$;
            \item[$(ii.2)$] If $d_1< d_1^*(r)$, then $(0,\tilde{w})$ is linearly stable for $a_1>a_1^*$ and unstable for $0<a_1<a_1^*$.
          \end{description}
    \end{description}

   Moreover,  $d_1^*(r)$ is continuous and strictly increases for $r\in (0,r^*)$ with $\lim\limits_{r\to{0}}d_1^*(r)=0$ and $\lim\limits_{r\to{r^*}}d_1^*(r)=+\infty.$
  \end{theorem}
  \begin{figure}[h]
  	\centering
  	\includegraphics[width=6cm,height=5cm]{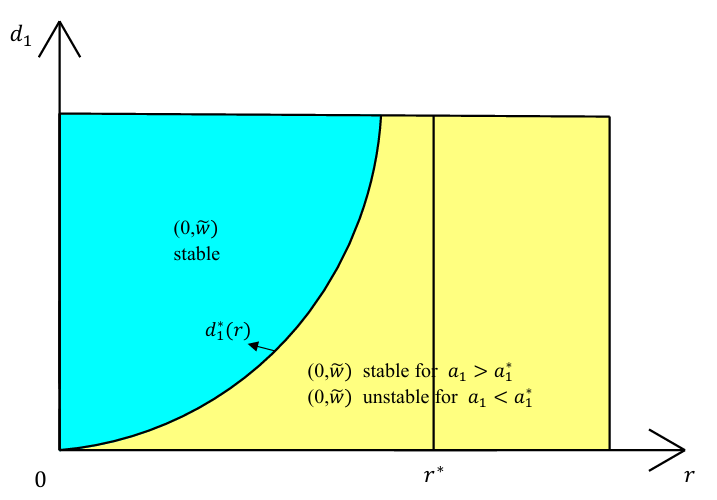}
  	\caption{ The effects of the advection rate $a_1$, the diffusion rate $d_1$, and the natural growth rate $r$ on the stability of $(0,\tilde{w})$. In the blue region, $(0,\tilde{w})$ is stable, implying the extinction of the population when rare. When $(r,d_1)$ lies in the yellow region, a small advection rate $(0<a_1<a_1^*)$ is beneficial for the population persistence.}
  	\label{fig:3}
  \end{figure}

Theorem \ref{th1.3} characterizes the effects of the advection rate $a_1$, the  diffusion coefficient $d_1$, and the  natural growth rate $r$ of population $u$ on its persistence (see Figure \ref{fig:3}). That is
\begin{itemize}
  \item High intrinsic growth rate $(r\ge r^*)$: For fixed diffusion coefficient $d_1$, population persistence becomes advection-dependent. Survival occurs at low advection rates ($0<a_1<a_1^*$), while extinction occurs at high advection rates ($a_1>a_1^*$).
  \item Low intrinsic growth rate ($r<r^*$): Extinction occurs for large diffusion coefficient ($d_1\ge d_1^*(r)$). If the diffusion coefficient is small ($d_1< d_1^*(r)$), a small advection rate $(0<a_1<a_1^*)$ is beneficial for the  population persistence.
\end{itemize}

  About the existence of a positive steady state $(u,w)$ for system~\eqref{eq1.1} and its stability, we have the following results.
  \begin{theorem}\label{th1.4}
  	Suppose that $0< a_1<\sqrt{4d_1r}$ and $L>L_1^*$.
  \begin{description}
    \item[$(i)$] For $x_1\ge L^*$ (large protection zone),  when $m\le p^{-1}e^{-|\frac{a_1}{d_1}-\frac{a_2}{d_2}|L}$ and $\tilde{w}^2(L)\le h$, system~\eqref{eq1.1} possesses a unique coexistence steady state $(u,w)$, which is globally asymptotically stable;
    \item[$(ii)$] For $x_1< L^*$ (small protection zone),  when $m<min\{m^*,p^{-1}e^{-|\frac{a_1}{d_1}-\frac{a_2}{d_2}|L}\}$ and $\tilde{w}^2(L)\le h$, system~\eqref{eq1.1} possesses a unique coexistence steady state $(u,w)$, which is globally asymptotically stable.
  \end{description}
  \end{theorem}

  Next, we give some main results for system \eqref{eq1.2} which are analogous to Theorem~\ref{th1.1}, \ref{th1.3} and \ref{th1.4}.
  About the influence of downstream protected zone length on the
stability of the equilibrium of $(0,\bar{w})$, we have the following results.
  \begin{theorem}\label{th1.5}
  If $a_1\ge\sqrt{4d_1r}$ or $0< a_1<\sqrt{4d_1r}$  and $L<L^*_1$, then $(0,\bar{w})$ is globally asymptotically stable; If $0< a_1<\sqrt{4d_1r}$ and $L>L_1^*$, then the following results hold:
  \begin{description}
    \item[$(i)$] If $L-x_2\ge L^*$ (large protection zone), then $(0,\bar{w})$ is unstable;
    \item[$(ii)$] If $L-x_2< L^*$ (small protection zone), then the following statements are true:
    \begin{description}
      \item[$(ii.1)$] If all parameters except for $h$ and $m$ are fixed, and $(h,m)\in \{h>0, m>\bar{m}(h)\}$, then $(0,\bar{w})$ is linearly stable; while for $(h,m)\in \{h>0, 0<m<\bar{m}(h)\}$,  $(0,\bar{w})$ is unstable.
      \item[$(ii.2)$] If all parameters except for $q$ and $m$ are fixed, and  $(q,m)\in\{0<q<\bar{q}(m)\text{,} m>0\}$, then $(0,\bar{w})$ is linearly stable; while for $(q,m)\in\{q>\bar{q}(m), m>0\}$,  $(0,\bar{w})$ is unstable.
    \end{description}
  \end{description}
  \end{theorem}

If the  protected zone is situated at the downstream of the river, Theorem \ref{th1.5} shows that the population also face extinction due to its large advection rate ($a_1\ge\sqrt{4d_1r}$) or small  advection rate while a short habitat length ($0< a_1<\sqrt{4d_1r}$, $L<L^*_1$). When the advection rate of the population is small ($0< a_1<\sqrt{4d_1r}$) and the  river
 length is large $(L>L_1^*)$, Theorem \ref{th1.5} characterizes the influence of downstream protected zone length on the stability of the equilibrium $(0,\bar{w})$. Specifically,
when the protected area exceeds a critical spatial threshold $(L-x_2\ge L^*)$, population persistence occurs irrespective of toxicant exposure levels;
   Below this threshold $(L-x_2<L^*)$, persistence becomes dependent on additional factors: the effect coefficient of toxicant on population growth $m$, per-unit contaminant discharge rate $q$, and toxicant input rate $h$.


For system~\eqref{eq1.2}, the following results show
 the effects of the advection rate $a_1$, the  diffusion coefficient $d_1$, and the  natural growth rate $r$ of population $u$ on its persistence.
  \begin{theorem}\label{th1.6}
  	Assume that all parameters except for $a_1$, $d_1$ and $r$ are fixed. Let $\bar{r}=\frac{m\int_{0}^{x_2}\bar{w}dx}{L}$. Then we have the following results:
  \begin{description}
    \item[$(i)$] If $r\ge \bar{r}$, given $d_1$, then there exists $\bar{a}_1>0$ satisfying $\lambda _1(d_1,\bar{a}_1,r-m\chi _{\left[ 0,x_2 \right)}\bar{w},\left( 0,L \right)) =0 $ such that $(0,\bar{w})$ is linearly stable when $a_1>\bar{a}_1$, and
    unstable when $0<a_1<\bar{a}_1$;
    \item[$(ii)$] If $r< \bar{r}$, then there exists $\bar{d}_1>0$ satisfying $\lambda _1\left( \bar{d}_1,0,r-m\chi _{\left[ 0,x_2 \right)}\bar{w},\left( 0,L \right) \right) =0$ such that the following statements are true:
        \begin{description}
          \item[$(ii.1)$] If $d_1 \ge \bar{d}_1(r)$, then $(0,\bar{w})$ is linearly stable for any $a_1>0$;
          \item[$(ii.2)$] If $d_1< \bar{d}_1(r)$, then $(0,\bar{w})$ is linearly stable for $a_1>\bar{a}_1$, and unstable for $0<a_1<\bar{a}_1$.
        \end{description}
  \end{description}

   Moreover,  $\bar{d}_1(r)$ is a  continuous function of $r$ and strictly decrease for $r\in(0,\bar{r})$ with $\lim\limits_{r\to{0}}\bar{d}_1(r)=0,\quad \lim\limits_{r\to{\bar{r}}}\bar{d}_1(r)=+\infty.$
  \end{theorem}
 Theorem \ref{th1.6} implies that if the intrinsic growth rate $r$ is sufficiently large, then a low advection rate $a_1$ promotes population persistence; If the intrinsic growth rate $r$ is small, then both a low diffusion rate $d_1$ and a low advection rate $a_1$
are required to ensure population persistence.


The following results show that, under some sufficient conditions, system~\eqref{eq1.2} admits a unique coexistence steady state, which is globally asymptotically stable.
   \begin{theorem}\label{th1.7}
  	Suppose that $0< a_1<\sqrt{4d_1r}$ and $L>L_1^*$. Then the following statements hold:
  \begin{description}
    \item[$(i)$] For $L-x_2\ge L^*$ (large protection zone), when $m\le p^{-1}e^{-|\frac{a_1}{d_1}-\frac{a_2}{d_2}|x_2}$ and $\bar{w}^2(x_2)\le h$, system~\eqref{eq1.2} possesses a unique coexistence steady state $(u,w)$ which is globally asymptotically stable;
    \item[$(ii)$] For $L-x_2< L^*$ (small protection zone), when $m<min\{\bar{m},p^{-1}e^{-|\frac{a_1}{d_1}-\frac{a_2}{d_2}|x_2}\}$ and $\bar{w}^2(x_2)\le h$, system~\eqref{eq1.2} possesses a unique coexistence steady state $(u,w)$ which is globally asymptotically stable.
  \end{description}
  \end{theorem}

  \section{Proof of the main results}
\subsection{Preliminaries}
First, we  consider the following associated stationary problems
  \begin{equation}\label{eq2.1}
  	\begin{cases}
  		d_2w_{xx}-a_2w_x+h-qw=0,&x_1<x<L,
  		\\
  		d_2w_x\left( x_1 \right) -a_2w\left( x_1 \right) =d_2w_x\left( L \right) =0,
  	\end{cases}
  \end{equation}
  and
   \begin{equation}\label{eq2.2}
  	\begin{cases}
  		d_2w_{xx}-a_2w_x+h-qw=0,&0<x<x_2,
  		\\
  		d_2w_x\left( 0 \right) -a_2w\left( 0 \right) =0,
  		\\
  		d_2w_x\left( x_2 \right) -a_2w\left( x_2 \right) =0.
  	\end{cases}
  \end{equation}
  and
  \begin{equation}\label{f1}
  	\begin{cases}
  		d_1u_{xx}-a_1u_x+u(r-u)=0,&0<x<L,
  		\\
  		d_1u_x\left( 0 \right) -a_1u\left( 0 \right) =d_1u_x\left( L \right) =0,
  	\end{cases}
  \end{equation}
 The following results about  the existence of positive solutions for problems $(\ref{eq2.1})$, $(\ref{eq2.2})$ and $(\ref{f1})$ can be easily obtained.
  	
  \begin{lemma}\label{lemma2.1}	
 $(\romannumeral 1)$  For any parameters $d_2,a_2,h,q>0$, problem \eqref{eq2.1} always has a unique positive solution, denoted by $\tilde{w}$, which can be expressed as $$\tilde{w}\left( x \right) =c_1e^{\alpha _1x}+c_2e^{\alpha _2x}+\frac{h}{q}$$
  with
  $$
  \alpha _1=\frac{a_2+\sqrt{a_{2}^{2}+4d_2q}}{2d_2},\quad
  \alpha _2=\frac{a_2-\sqrt{a_{2}^{2}+4d_2q}}{2d_2},
  $$
  $$
  c_1=\frac{ha_2\alpha _2}{q\left[ \alpha _2\left( d_2\alpha _1-a_2 \right) e^{\alpha _1x_1}-\alpha _1e^{L\left( \alpha _1-\alpha _2 \right)}\left( d_2\alpha _2-a_2 \right) e^{\alpha _2x_1} \right]}
  , \quad
  c_2=-\frac{\alpha _1}{\alpha _2}e^{L\left( \alpha _1-\alpha _2 \right)}c_1.
  $$

  $(\romannumeral 2)$  For any parameters $d_2,a_2,h,q>0$, problem \eqref{eq2.2} always has a unique positive solution, denoted by $\bar{w}$, which can be expressed as $$\bar{w}\left( x \right) =k_1e^{\alpha _1x}+k_2e^{\alpha _2x}+\frac{h}{q}.$$
  Here
  $$
  \alpha _1=\frac{a_2+\sqrt{a_{2}^{2}+4d_2q}}{2d_2}
  ,\quad
  \alpha _2=\frac{a_2-\sqrt{a_{2}^{2}+4d_2q}}{2d_2},
  $$
  $$
  k_1=\frac{ha_2\left( 1-e^{\alpha _2x_2} \right)}{q\left( d_2\alpha _1-a_2 \right) \left( e^{\alpha _1x_2}-e^{\alpha _2x_2} \right)}
  ,\quad
  k_2=\frac{ha_2\left( 1-e^{\alpha _1x_2} \right)}{q\left( d_2\alpha _2-a_2 \right) \left( e^{\alpha _2x_2}-e^{\alpha _1x_2} \right)}.
  $$

 $(\romannumeral 3)$\cite{lou2015evolution} For any parameters $d_1,r>0$, there exists $L_1^*$ such that problem \eqref{f1} has a unique positive solution, denoted by $\tilde{u}$, if $0< a_1<\sqrt{4d_1r}$ and $L>L_1^*$. If $a_1\ge \sqrt{4d_1r}$ or $0< a_1< \sqrt{4d_1r}$ and $L<L_1^*$, then $0$ is the unique nonnegative solution to problem \eqref{f1}.
  \end{lemma}

  \begin{lemma}\label{lemma 2.2}
   The following statements on $\tilde{w}$ are true.

  $(\romannumeral 1)$  $0<\tilde{w}_x<\frac{a_2}{d_2}\tilde{w}$ in $(x_1,L)$.

   $(\romannumeral 2)$  There exists $\delta\in(0,1)$ such that $\delta\frac{h}{q}\le\tilde{w} <\frac{h}{q}$ on $[x_1,L]$.

 $(\romannumeral 3)$ $\tilde{w}$ is strictly increasing with respect to $h$.

 $(\romannumeral 4)$ $\tilde{w}$ is strictly decreasing with respect to $a_2$.

 $(\romannumeral 5)$  $\lim\limits_{h \to 0} \tilde{w} = 0, \quad \lim\limits_{h \to +\infty} \tilde{w} = +\infty$
   uniformly on $[x_1,L]$.

 $(\romannumeral 6)$  $\lim\limits_{a_2 \to 0} \tilde{w} = \frac{h}{q}, \quad \lim\limits_{a_2 \to +\infty} \tilde{w} = 0$ uniformly on $[x_1,L]$.

 $(\romannumeral 7)$ $\tilde{w}$ is strictly decreasing with respect to $q$, and $\lim\limits_{q \to 0} \tilde{w} = +\infty, \quad \lim\limits_{q \to +\infty} \tilde{w} = 0$ uniformly on $[x_1,L]$.
  \end{lemma}
  \begin{proof}
  $(\romannumeral 1)$ Set $V=\frac{\tilde{w}_x}{\tilde{w}}$. Then $V$ satisfies
  \begin{equation*}
  	\begin{cases}
  	-d_2V_{xx}+(a_2-2d_2V)V_x+\frac{h}{\tilde{w}}V=0,&x_1<x<L,\\
  	V(x_1)=\frac{a_2}{d_2}, V(L)=0.
  	\end{cases}
  \end{equation*}
  The desired result follows from the strong maximum principle \cite{evans2010pde} for elliptic equations.

 $(\romannumeral 2)$ Through simple calculations, it can be concluded that $c_1,c_2<0$, which indicates that
\begin{equation}\label{sm}
\tilde{w}(x)<\frac{h}{q}.
\end{equation}
From $(\romannumeral 1)$, we have
 \begin{equation}\label{eq2.4}
 \tilde{w}(x_1)\ge \tilde{w}(L)e^{-\frac{a_2}{d_2}(L-x_1)}.
 \end{equation}
  Integrating the first equation of $(\ref{eq2.1})$ on $[x_1,L]$, with the help of the boundary conditions, we find that
  \begin{equation}\label{eq2.5}
 \tilde{w}(L)\ge \frac{h(L-x_1)}{a_2+q(L-x_1)}=h\frac{1}{\frac{a_2}{L-x_1}+q}.
  \end{equation}
  From \eqref{eq2.4} and \eqref{eq2.5}, we have
  \begin{equation}\label{eq2.6}
  \tilde{w}(x)\ge \tilde{w}(x_1)\ge h\frac{1}{\frac{a_2}{L-x_1}+q}e^{-\frac{a_2}{d_2}(L-x_1)}.
  \end{equation}
  Notice that
  $$0<\frac{1}{\frac{a_2}{L-x_1}+q}e^{-\frac{a_2}{d_2}(L-x_1)}<\frac{1}{q}.$$
  Then there exists some $\delta \in(0,1)$
  such that
  \begin{equation}\label{eq2.7}
  	\delta\frac{1}{q}=\frac{1}{\frac{a_2}{L-x_1}+q}e^{-\frac{a_2}{d_2}(L-x_1)}.
  \end{equation}
  From \eqref{eq2.6} and \eqref{eq2.7}, we have
 \begin{equation}\label{bm}
  \tilde{w}(x)\ge \delta\frac{h}{q}\quad x\in[x_1,L].
  \end{equation}
  From \eqref{sm} and \eqref{bm} we get $(\romannumeral 2)$.

  $(\romannumeral 3)$
  Setting $\acute{\tilde{w}}=\frac{\partial \tilde{w}}{\partial h}$, we have
  \begin{equation*}
  	\begin{cases}
  	-d_2\acute{\tilde{w}}_{xx}+a_2\acute{\tilde{w}}_x+q\acute{\tilde{w}}=1>0,&x_1<x<L,
  	\\
  	-d_2\acute{\tilde{w}}_x(x_1)+a_2\acute{\tilde{w}}(x_1)=\acute{\tilde{w}}_x(L)=0.
  	\end{cases}
  \end{equation*}
  It then follows from the strong maximum principle \cite{evans2010pde} for elliptic equations that $\acute{\tilde{w}}>0$ for $x\in(x_1,L)$, which implies that $(\romannumeral 3)$ holds.

$(\romannumeral 4)$
  Set $\grave{\tilde{w}}=\frac{\partial \tilde{w}}{\partial a_2}$. Then $\grave{\tilde{w}}$ satisfies
  \begin{equation*}
  \begin{cases}
  	-d_2\grave{\tilde{w}}_{xx}+a_2\grave{\tilde{w}}_x+q\grave{\tilde{w}}=-\tilde{w}_x<0,&x_1<x<L,
  	\\
  	-d_2\grave{\tilde{w}}_x(x_1)+a_2\grave{\tilde{w}}(x_1)=-\tilde{w}(x_1)<0,
  	\\
  	\grave{\tilde{w}}_x(L)=0.
  \end{cases}
  \end{equation*}
  It then follows from the strong maximum principle for elliptic equations that $\grave{\tilde{w}}<0$ for $x\in(x_1,L)$, which implies that $(\romannumeral 4)$ holds.

  $(\romannumeral 5)$ In light of the Squeeze Theorem \cite{rudin1976}, together with $(\romannumeral 2)$, we can get $(\romannumeral 5)$ easily.

 $(\romannumeral 6)$ From $Lemma~\ref{lemma2.1}$, we see that $\tilde{w}$ is continuously differentiable with respect to $a_2$ for $x\in(x_1,L)$. It is easy to see that as $a_2\rightarrow 0$, $\tilde{w}$ converges to $\hat{w}$, where $\hat{w}$ fulfills
 \begin{equation}\label{s}
 	\begin{cases}
 		d_2\hat{w}_{xx}+h-q\hat{w}=0 ,&x_1<x<L,
 		\\
 		\hat{w}_x(x_1)=\hat{w}_x(L)=0.
 	\end{cases}
 \end{equation}
Solving \eqref{s}, we obtain $\hat{w}=\frac{h}{q}$. Hence, $\lim\limits_{a_2 \to 0} \tilde{w} = \frac{h}{q}$. Due to $\tilde{w}_x>0$ for $x\in(x_1,L)$, we have $0<\tilde{w}<\tilde{w}(L)$. Then to prove $\lim\limits_{a_2 \to +\infty} \tilde{w} = 0$, we only need to verify that $\lim\limits_{a_2 \to +\infty} \tilde{w}(L) = 0$. Integrating the first equation of $(\ref{eq2.1})$ on $[x_1,L]$, we obtain $a_2\tilde{w}(L)=h(L-x_1)-q\int_{x_1}^L{\tilde{w}}dx<h(L-x_1)$. Therefore $0\le \tilde{w}(L)<\frac{h(L-x_1)}{a_2}\rightarrow 0$ as $a_2\rightarrow+\infty$, which implies that $\lim\limits_{a_2 \to +\infty} \tilde{w}(L) = 0$.

$(\romannumeral 7)$ The proof is similar to that of $(\romannumeral 3)$ and $(\romannumeral 5)$.
  \end{proof}
\begin{lemma}\label{lemma2.3}
	 The following statements on $\bar{w}$ are true.

 $(\romannumeral 1)$ $0<\bar{w}_x<\frac{a_2}{d_2}\bar{w}$ in $(0,x_2)$.

 $(\romannumeral 2)$ $\bar{w}$ is strictly increases with respect to $h$.

  $(\romannumeral 3)$ $\lim\limits_{h \to 0} \bar{w} = 0, \quad \lim\limits_{h \to +\infty} \bar{w} = +\infty$
  uniformly on $[0,x_2]$.

   $(\romannumeral 4)$ $\bar{w}$ is strictly decreasing with respect to $q$.

  $(\romannumeral 5)$ $\lim\limits_{q \to 0} \bar{w} = +\infty, \quad \lim\limits_{q \to +\infty} \bar{w} = 0$
  uniformly on $[0,x_2]$.
\end{lemma}
\begin{proof}
The  proofs of $(\romannumeral 1)$, $(\romannumeral 4)$ and $(\romannumeral 5)$ are similar to Lemma~\ref{lemma 2.2}. Statement $(\romannumeral 3)$ directly follows from the expression of $\bar{w}(x)$ and the continuity of $\bar{w}(x)$ on $h$. Then we only prove $(\romannumeral 2)$.

	Denote $\mu=e^{-\frac{a_2}{d_2}x}\bar{w}$. Then $\mu$ satisfies
	\begin{equation*}
	 \begin{cases}
	 	-d_2\mu_{xx}-a_2\mu_x+q\mu-he^{\frac{a_2}{d_2}x}=0, &0<x<x_2,\\
	 	\mu_x(0)=\mu_x(x_2)=0.
	 \end{cases}
	\end{equation*}
	Setting $\dot{\mu}=\frac{\partial \mu}{\partial h}$, similar to the proof of $Lemma~\ref{lemma 2.2} (\romannumeral 3)$ and utilizing the strong maximum principle, we obtain $\dot{\mu}>0$ for $x\in(0,x_2)$, which implies that $(\romannumeral 2)$ holds.
\end{proof}

Next, for our further discussion, we recall some results on the
following eigenvalue problem:
  \begin{equation}\label{eq2.12}
  	\begin{cases}
  	d\varphi _{xx}-a\varphi_{x}+s(x)\varphi=\lambda\varphi ,&y_1<x<y_2,\\
  	d\varphi_x(y_1)-a\varphi(y_1)=\varphi_x(y_2)=0,
  	\end{cases}
  \end{equation}
  where $0\le y_1<y_2\le L$ and $s(x)\in L^\infty(y_1,y_2)$. Denoted the principal eigen-pair of problem $(\ref{eq2.12})$ by $(\lambda_1(d,a,s(x),(y_1,y_2)),\varphi_1(d,a,s(x),(y_1,y_2)))$. The principal eigenvalues of $(\ref{eq2.12})$ can be characterized by
 \begin{equation}\label{eq2.13}
 	\lambda_1(d,a,s(x),(y_1,y_2))=\sup_{0\ne \varphi \in H^1\left( y_1,y_2 \right)}\frac{\int_{y_1}^{y_2}{\left( -d\varphi _{x}^{2}e^{\frac{a}{d}x}+s(x)\varphi ^2e^{\frac{a}{d}x} \right)}dx-a\varphi ^2\left( y_2 \right) e^{\frac{a}{d}y_2}}{\int_{y_1}^{y_2}{\varphi ^2e^{\frac{a}{d}x}}dx}.
 \end{equation}
 It is well know that $\lambda_1(d,a,s(x),(y_1,y_2))$ is simple and its corresponding eigenfunction $\varphi_1(d,a,s(x),(y_1,y_2))$ can be chosen strictly positive in $(y_1,y_2)$. From \cite[Lemma2.1]{bushi2023} and \cite[Theorem 1.2]{peng2019},  we have the following properties for  the principal eigenvalue of problem $(\ref{eq2.12})$.
 \begin{lemma}\label{lemma2.4}
 	Suppose that $s(x)$ is a continuous function on $[y_1,y_2]$ for any $0\le y_1<y_2\le L$. Then the
 	following statements concerning $\lambda_1(d,a,s(x),(y_1,y_2))$ are valid.
 	
 $(\romannumeral 1)$ $\lambda_1(d,a,s(x),(y_1,y_2))$ is continuously differentiable with respect to  $a$ and $d$, respectively.

 $(\romannumeral 2)$ If $s(x)$ isn't a constant function in $(y_1,y_2)$, then
 $$\lim\limits_{d\to 0} \lambda_1(d,0,s(x),(y_1,y_2))=\underset{y\in [y_1,y_2]}{max} s(x) ,\quad  \lim\limits_{d\to{+\infty}} \lambda_1(d,a,s(x),(y_1,y_2))=\frac{\int_{y_1}^{y_2} s(x) dx}{y_2-y_1}.$$

 $(\romannumeral 3)$ $\lambda_1(d,0,s(x),(y_1,y_2))$ is strictly decreasing with respect to $d\in (0,+\infty)$.

 $(\romannumeral 4)$ $\lambda_1(d,a,s(x),(y_1,y_2))$ is strictly decreasing with respect to $a\in (0,+\infty)$, and $\lim\limits_{a\to{+\infty}} \lambda_1(d,a,s(x),(y_1,y_2))=-\infty$.

 $(\romannumeral 5)$ If $s_i(x)\in C([0,L])(i=1,2)$ and $s_1(x)\ge,\not\equiv s_2(x) $ in $(y_1,y_2)$, then $$\lambda_1(d,a,s_1(x),(y_1,y_2))>\lambda_1(d,a,s_2(x),(y_1,y_2)).$$
 \end{lemma}

 To support the proofs by means of monotone semiflow theory in  next subsection, we establish an order structure (including cones and partial order relations). Let $X := C([0, L])$ and $X^+$ be the collection of non-negative continuous functions on $[0, L]$. For the intervals $[x_1, L]$ and $[0, x_2]$, we introduce $X_1^+$ and $X_2^+$ as the sets of non-negative continuous functions thereon, respectively. Define cones $K_i := X^+ \times (-X_i^+)$ ($i = 1, 2$), where the interior of $K_i$ (non-empty by construction) is given by $\text{Int}K_i = \text{Int}X^+ \times (-\text{Int}X_i^+)$. The partial order relations $\leq_{K_i}$, $<_{K_i}$, and $\ll_{K_i}$ ($i = 1, 2$) correspond to the standard orderings induced by $K_i$, its punctured version $K_i \setminus \{(0, 0)\}$, and the interior $\text{Int}K_i$, respectively. More precisely, for $u_1, u_2 \in X^+$ and $w_1, w_2\in X_i^+$ ($i = 1, 2$)
 \begin{equation*}
 \begin{aligned}
&(u_1, w_1) \leq_{K_i} (u_2, w_2) \Longleftrightarrow u_1 \leq u_2 \text{ in } [0, L] \text{ and } w_2 \leq w_1 \text{ in } \Omega_i,\\
&(u_1, w_1) <_{K_i} (u_2, w_2) \Longleftrightarrow (u_1, w_1) \leq_{K_i} (u_2, w_2) \text{ and } \left( u_1 \neq u_2 \text{ in } [0, L] \text{ or } w_1 \neq w_2 \text{ in } \Omega_i \right),\\
&(u_1, w_1) \ll_{K_i} (u_2, w_2) \Longleftrightarrow u_1 < u_2 \text{ in } [0, L] \text{ and } w_2 < w_1 \text{ in } \Omega_i,
\end{aligned}
\end{equation*}
 where $\Omega_1 = [x_1, L]$ and $\Omega_2 = [0, x_2]$.

\subsection{Dynamics of system (\ref{eq1.1})}
 From Lemma \ref{lemma2.1}, system \eqref{eq1.1} admits a semi-trivial steady state $(0,\tilde{w})$. In this subsection, we will investigate the local dynamics of the semi-trivial steady state $(0,\tilde{w})$ and the coexistence steady state $(u,w)$ of system \eqref{eq1.1}.

 It is well-known that the linear stability of $(0,\tilde{w})$ is determined by the following eigenvalue problem
 \begin{equation*}
 	\begin{cases}
 		d_1\phi_{xx}-a_1\phi_x+(r-m\chi _{\left( x_1,L \right]}\tilde{w})\phi=\lambda\phi,&x\in (0,L),\\
 		d_1\phi_x(0)-a_1\phi(0)=\phi_x(L)=0.
 	\end{cases}
 \end{equation*}
 Moreover, similarly to \cite[Corollary 2.10]{sweers1992strong}, we have
 \begin{equation*}
 	if \quad \lambda_1(d_1,a_1,r-m\chi _{\left( x_1,L \right]}\tilde{w},(0,L))
 	\begin{cases}
 		<0,\\
 		=0,\\
 		>0,
 	\end{cases}
 	\quad
 	\text{then }(0,\tilde{w}) \text{ is}
 	\begin{cases}
 		\text{linearly stable},\\
 		\text{neutrally stable},\\
 		\text{linearly unstable}.
 	\end{cases}
 \end{equation*}

 Next, we introduce the following eigenvalue problem
 \begin{equation}\label{eq3.2}
 	\begin{cases}
 		d\varphi_{xx}-a\varphi_x+r\varphi=\lambda^D\varphi,& y_1<x<y_2,\\
 		d\varphi_x(y_1)-a\varphi(y_1)=\varphi(y_2)=0.
 	\end{cases}
 \end{equation}
 Here, $0\le y_1<y_2\le L$. The principal eigen-pair of problem $(\ref{eq3.2})$ is defined as $(\lambda_{1}^D(d,a,r,(y_1,y_2)),\varphi_1^D(d,a,r,(y_1,y_2)))$. According to \cite{speirs2001population}, the following results hold.
 \begin{proposition}\label{propB}
 	Assume that $0<a_1<\sqrt{4d_1r}$.
 \begin{equation}\label{eq3.3}
 If \quad	x_1
 	\begin{cases}
 		< L^*, \\
 		= L^*,\\
 		> L^*,
 	\end{cases}
 	\quad
 	\text{then }\quad \lambda_{1}^{D}(d_1, a_1, r, (0, x_1))
 	\begin{cases}
 		 < 0, \\
 		 = 0, \\
 		 > 0.
 	\end{cases}
 \end{equation}
 \end{proposition}
 \begin{lemma}\label{lemma3.1}
 	$\lim\limits_{\xi \to{+\infty}} \lambda_1(d,a,r-\xi\chi_{\left( x_1,L \right]},(0,L))=\lambda_1^D(d,a,r,(0,x_1)).
 $
 \end{lemma}
 \begin{proof}
 	With the help of \eqref{eq2.13}, we have
 	\begin{equation*}
 		\begin{split}
 	&\lambda_1(d,a,r-\xi\chi_{\left( x_1,L \right]},(0,L))\\
  	&= \sup_{0\ne \varphi \in H^1\left( 0,L \right)}\frac{\int_{0}^{L}{\left( -d\varphi _{x}^{2}e^{\frac{a}{d}x}+(r-\xi\chi_{\left( x_1,L \right]})\varphi ^2e^{\frac{a}{d}x} \right)}dx-a\varphi ^2\left( L \right) e^{\frac{a}{d}L}}{\int_{0}^{L}{\varphi ^2e^{\frac{a}{d}x}}dx}\\
 	&\ge \frac{\int_0^L{\left( -d\left( \tilde{\varphi} \right) _{x}^{2}e^{\frac{a}{d}x}+r\left( \tilde{\varphi} \right) ^2e^{\frac{a}{d}x} \right)}dx}{\int_0^L{\left( \tilde{\varphi} \right) ^2e^{\frac{a}{d}x}}dx}\\
 	&=\lambda_1^D(d,a,r,(0,x_1)),
 		\end{split}
 	\end{equation*}
 	where
 	\begin{equation*}
 		\tilde{\varphi}=
 		\begin{cases}
 		\varphi ^D_1(d,a,r,(0,x_1)),& x\in [0,x_1],\\
 		0,& x\in [x_1,L].
 		\end{cases}
 	\end{equation*}

 	Define  $R(\xi)=\lambda_1(d,a,r-\xi\chi_{\left( x_1,L \right]},(0,L))$. From  Lemma~\ref{lemma2.4}$(\romannumeral 5)$,  $R(\xi)$ is continuous and monotonically decreases with respect to $\xi$. Let the sequence $\{\varphi_n\}$ $(\varphi_n>0)$ satisfy
 	\begin{equation}\label{eq3.5}
 		\begin{cases}
 			d(\varphi_n) _{xx}-a(\varphi_n)_{x}+(r-\xi_n\chi_{\left( x_1,L \right]})\varphi_n=R(\xi_n)\varphi_n ,&0<x<L,\\
 			d(\varphi_n)_x(0)-a(\varphi_n)(0)=(\varphi_n)_x(L)=0,
 		\end{cases}
 	\end{equation}
 	where $\left\{ \xi _n \right\} $ is a sequence satisfying $\xi_n\to{+\infty}$ $(n\to{+\infty})$, and assume that $\lVert \left. \varphi _n \rVert \right. _{\infty}=1$.
 	Multiplying the first equation of $(\ref{eq3.5})$ by $\varphi_n$ and integrating  the resulting equation over $(0,L)$, we obtain
 \begin{equation}\label{eq3.6}
 	\begin{aligned}
 		 \int_{0}^{L}d (\varphi_n)_x^2 dx
 		 &=-a(\varphi_n)^2(L)-\xi_n\int_{x_1}^{L}(\varphi_n)^2dx-R(\xi_n)\int_{0}^{L} (\varphi_n)^2dx\\	&\quad+r\int_{0}^{L} (\varphi_n)^2dx
 	 +a\int_{0}^{L} (\varphi_n)_x \varphi_ndx \\
 		&\le r\int_{0}^{L} (\varphi_n)^2dx+\frac{a^2}{2d}\int_{0}^{L} (\varphi_n)^2dx+\frac{d}{2}\int_{0}^{L} (\varphi_n)^2_xdx\\
 		&\quad-\lambda_1^D(d,a,r,(0,x_1))\int_{0}^{L} (\varphi_n)^2_xdx.
 	\end{aligned}
 \end{equation}
 	From \eqref{eq3.6} we have
 	\begin{align*}
 		\frac{d}{2} \int_{0}^{L}(\varphi_n)_x^2dx&\le (r+\frac{a^2}{2d}- \lambda_1^D(d,a,r,(0,x_1)) )\int_{0}^{L} (\varphi_n)^2dx\\
 		&\le (r+\frac{a^2}{2d}- \lambda_1^D(d,a,r,(0,x_1)) )L.
 	\end{align*}
 	Therefore, $\left\{ \varphi _n \right\} $ is bounded in $H^1([0,L])$, and then there exists a subsequence (still denoted by $\left\{ \varphi _n \right\} $) converging to some $\hat{\varphi}$ weakly in $H^1([0,L])$ and strongly in $L^2([0,L])$. Clearly, $\hat{\varphi}\ge 0$ in $[0,L]$
 and $\lVert \left. \hat{\varphi} \rVert \right. _{\infty}=1$.
 	
 	Integrating the first equation of \eqref{eq3.5} over $(0,L)$, one has
 	\begin{equation}\label{eq3.7}
 	\begin{aligned}
 		\xi_n\int_{x_1}^{L}\varphi_ndx
 		&=-a\varphi_n(L)-R(\xi_n)\int_{0}^{L}\varphi_ndx+r\int_{0}^{L}\varphi_ndx\\
 		&\le (r- \lambda_1^D(d,a,r,(0,x_1)) )L.
 	\end{aligned}
 	\end{equation}
 	As $n\to{+\infty}$, the right-hand side of $\eqref{eq3.7}$ is bounded, but $\xi_n\to{+\infty}$ $(n\to{+\infty})$, so it is necessary that $\int_{x_1}^{L}\varphi_ndx\to 0$, which implies that $\hat{\varphi}\equiv0$ in $[x_1,L]$. On the other hand, since $R(\xi)$ is bounded below and monotonically decreasing with respect to $\xi$,  we may assume that $R(\xi_n)\to \hat{\lambda}$$(n\to{+\infty})$. Thus letting $n\to{+\infty}$ in equation $\eqref{eq3.5}$, we obtain $\hat{\varphi}$ and $\hat{\lambda}$ satisfy
 	\begin{equation}
 		\begin{cases}
 		d\hat{\varphi}_{xx}-a\hat{\varphi}_x+r\hat{\varphi}=\hat{\lambda}\hat{\varphi},&x\in(0,x_1),\\
 		d\hat{\varphi}_x(0)-a\hat{\varphi}(0)=\hat{\varphi}(x_1)=0.
 		\end{cases}
 	\end{equation}
 	Obviously, $\hat{\varphi} \not\equiv 0$ in $(0,x_1)$,  for otherwise $\hat{\varphi}\equiv0$ in $(0,L)$ which contradicts to $\lVert \left. \hat{\varphi} \rVert \right. _{\infty}=1$.  Thus, we conclude that $\hat{\lambda}=\lambda_1^D(d,a,r,(0,x_1))$. The proof is completed.
 \end{proof}

 Next, to prove Theorem~\ref{th1.1} and Theorem~\ref{th1.2}, we first establish several lemmas.

\begin{lemma}\label{l1}
	The following statements on system \eqref{eq1.1} are true.
	
	$(\romannumeral 1)$ If $\tilde{u}$ does not exist, then $(0,\tilde{w})$ is globally asymptotically stable.
	
	$(\romannumeral 2)$ If $\tilde{u}$ exists, then system~\eqref{eq1.1} has a nonnegative steady state $(u,w)$  satisfying $(0,\tilde{w})\leq_{K_1} (u,w) \leq_{K_1} (\tilde{u},\underaccent{\sim}{w})$, where $\underaccent{\sim}{w}$ is the unique positive solution to
	 \begin{equation}\label{f2}
		\begin{cases}
			d_2w_{xx}-a_2w_x+h-qw-p\tilde{u}w=0 ,&x\in \left( x_1,L \right),
			\\
			d_2w_x\left( x_1 \right) -a_2w\left( x_1\right) =w_x\left( L \right) =0.
		\end{cases}
	\end{equation}
\end{lemma}
\begin{proof}
	$(\romannumeral 1)$ we consider an auxiliary system:
	\begin{equation}\label{f3}
		\begin{cases}
			v_t=d_1v_{xx}-a_1v_x+v(r-v) ,&x\in \left( 0,L \right),
			\\
			d_1v_x\left( 0 \right) -a_1w\left( 0\right) =v_x\left( L \right) =0,
			\\
			v(x,0)\ge u(x,0).
		\end{cases}
	\end{equation}
	Since $\tilde{u}$ does not exist, one can verify that for any solution $v(x,t)$ in problem~\eqref{f3} satisfy
$$v(x,t) \to 0 \quad \text{uniformly} \quad \text{for} \quad x \in [0,L], \quad \text{as} \quad t \to \infty.$$
By the comparison principle for parabolic equations,  it can be inferred that for any given $(u_0,w_0)$ in system~\eqref{eq1.1},  the corresponding positive solution $(u(x,t),w(x,t))$ satisfies $u(x,t)\leq v(x,t)$ for any $t>0$. Hence,
$$u(x,t) \to 0 \quad \text{uniformly} \quad \text{for} \quad x \in [0,L], \quad \text{as} \quad t \to \infty.$$
This can be further shown that
$$w(x,t) \to \tilde{w} \quad \text{uniformly} \quad \text{for} \quad x \in [x_1,L], \quad \text{as} \quad t \to \infty.$$

$(\romannumeral 2)$ $\tilde{w}$ and $0$ are respectively the upper and lower solutions of problem~\eqref{f2}. Furthermore, by the method of upper and lower solutions for elliptic equations, it follows that $0<\underaccent{\sim}{w}<\tilde{w}$.

Consider the steady state equation corresponding to system~\eqref{eq1.1}
\begin{equation}\label{3.13}
	\begin{cases}
		d_1u_{xx}-a_1u_x+u\left( r-u-m\chi _{\left( x_1,L \right]}w \right)=0 ,&x\in \left( 0,L \right),
		\\
		d_2w_{xx}-a_2w_x+h-qw-puw=0 ,&x\in \left( x_1,L \right),
		\\
		d_1u_x\left( 0 \right) -a_1u\left( 0 \right) =u_x\left( L \right) =0,
		\\
		d_2w_x\left( x_1 \right) -a_2w\left( x_1\right) =w_x\left( L \right) =0.
	\end{cases}
\end{equation}
 According to Pao \cite{pao2012nonlinear}, $(\tilde{u},\underaccent{\sim}{w})$ and $(0,\tilde{w})$ are a pair of ordered upper and lower solutions of problem~\eqref{3.13}. Furthermore, since $(0,\tilde{w})<_{k_1} (\tilde{u},\underaccent{\sim}{w})$, by the comparison principle for upper and lower solutions, problem~\eqref{3.13} has a nonnegative solution $(u,w)$ satisfying
 $$(0,\tilde{w})\leq_{k_1} (u,w) \leq_{k_1} (\tilde{u},\underaccent{\sim}{w}).$$
\end{proof}
\begin{lemma}\label{lemma3.2}
	If $a_1\ge \sqrt{4d_1r}$ or $0< a_1< \sqrt{4d_1r}$ and $L<L_1^*$, then $(0,\tilde{w})$ is linearly stable. If $0< a_1<\sqrt{4d_1r}$ and $L>L_1^*$, then the following results on its linear stability hold:

	$(\romannumeral 1)$ If $x_1\ge L^*$ $(\text{large protection zone})$, then $(0,\tilde{w})$ is unstable;

	$(\romannumeral 2)$ If $x_1<L^*$ $(\text{small protection zone})$, fixed all parameters except for $x_1$ and $m$, then there exists $m^*(x_1)>0$ satisfying $\lambda _1\left( d_1,a_1,r-m\chi _{\left( x_1,L \right]}\tilde{w},\left( 0,L \right) \right) =0$ such that the following statements hold:

	$\quad$$(\romannumeral 2.1)$ If $m>m^*(x_1)$, then $(0,\tilde{w})$ is linearly stable;

	$\quad$$(\romannumeral 2.2)$ If $0<m<m^*(x_1)$, then $(0,\tilde{w})$ is unstable.
	
	Furthermore, $m^*(x_1)$ is a continuous function of $x_1$ and is strictly increasing on $(0,L)$, with $\lim\limits_{x_1\to 0}m^*(x_1)=m^*(0)>0$ and $\lim\limits_{x_1\to L^*} m^*(x_1)=+\infty$.
\end{lemma}
\begin{proof}
From \cite[Lemma 2.2]{lou2015evolution},  $\lambda _1\left( d_1,a_1,r,\left( 0,L \right) \right)<0$ for any $L>0$  when $a_1\ge\sqrt{4d_1r}$, and
	when $0< a_1<\sqrt{4d_1r}$,
	\begin{equation*}
		\lambda _1\left( d_1,a_1,r,\left( 0,L \right) \right)
		\begin{cases}
			>0 ,&L>L_1^*,\\
			=0 ,&L=L_1^*,\\
			<0 ,&L<L_1^*.
		\end{cases}
	\end{equation*}

	Due to $r-m\chi _{\left( x_1,L \right]}\tilde{w}\le,\not\equiv r$, by Lemma~\ref{lemma2.4}$(\romannumeral 5)$, we obtain that for $a_1\ge \sqrt{4d_1r}$ or $0< a_1<\sqrt{4d_1r}$ and $L<L_1^*$,
$$\lambda _1\left( d_1,a_1,r-m\chi _{\left( x_1,L \right]}\tilde{w},\left( 0,L \right) \right)<\lambda _1\left( d_1,a_1,r,\left( 0,L \right) \right)<0,$$
which implies that $(0,\tilde{w})$ is linearly stable.

	For the case $0<a_1<\sqrt{4d_1r}$ and $L>L_1^*$, we have $\lambda _1\left( d_1,a_1,r,\left( 0,L \right) \right)>0$. From Lemma~$\ref{lemma3.1}$, we have
\begin{equation}\label{eq3.9}
	\lim\limits_{m\to{+\infty}}\lambda _1\text{(}d_1,a_1,r-m\chi _{\left( x_1,L \right]}\tilde{w},\left( 0,L \right) \text{)}=\lambda _{1}^{D}\left( d_1,a_1,r,\left( 0,x_1 \right) \right).
\end{equation}
	This, combined with the fact that $\lambda _1\text{(}d_1,a_1,r-m\chi _{\left( x_1,L \right]}\tilde{w},\left( 0,L \right) \text{)}$ is strictly decreasing with respect to $m$, implies that
\begin{align*}
		\lambda _1\text{(}d_1,a_1,r-m\chi _{\left( x_1,L \right]}\tilde{w},\left( 0,L \right) \text{)}&>
		\lim\limits_{m\to{+\infty}}\lambda _1\text{(}d_1,a_1,r-m\chi _{\left( x_1,L \right]}\tilde{w},\left( 0,L \right) \text{)}\\&=\lambda _{1}^{D}\left( d_1,a_1,r,\left( 0,x_1 \right) \right).
\end{align*}
 This, combined with Proposition \ref{propB}, shows that $\lambda _1\text{(}d_1,a_1,r-m\chi _{\left( x_1,L \right]}\tilde{w},\left( 0,L \right) \text{)}>0$ when $x_1\ge L^*$, which implies that  $(0,\tilde{w})$ is unstable.

For the case $x_1< L^*$, combining $\lambda _1\text{(}d_1,a_1,r,\left( 0,L \right) \text{)}>0$ with \eqref{eq3.9} and Proposition~\ref{propB}, and noting that $\lambda _1\text{(}d_1,a_1,r-m\chi _{\left( x_1,L \right]}\tilde{w},\left( 0,L \right) \text{)}$ is monotonically decreasing with respect to $m$, we obtain that there exists $m^*(x_1)$ such that
	\begin{equation*}
	\lambda _1\text{(}d_1,a_1,r-m\chi _{\left( x_1,L \right]}\tilde{w},\left( 0,L \right) \text{)}
		\begin{cases}
			>0, \quad m<m^*(x_1),\\
			=0, \quad m=m^*(x_1),\\
			<0, \quad m>m^*(x_1).
		\end{cases}
	\end{equation*}
	
	Next, we investigate the monotonicity of $m^*(x_1)$ with respect to $x_1$.
	Suppose that $0<x_{11}<x_{12}<L$. Let $\tilde{w}_i$ be the unique solution of \eqref{eq2.1} with $x_1=x_{1i}$. From the expression of $\tilde{w}$, it is easy to know that $\tilde{w}$ is strictly decreasing with respect to $x_1$. Thus $\tilde{w}_1(x)>\tilde{w}_2(x)$ on $[x_{11},L]$. Then $r-m\chi _{\left( x_{12},L \right]}\tilde{w}_2\ge,\not\equiv r-m\chi _{\left( x_{11},L \right]}\tilde{w}_1$. Combining with Lemma~\ref{lemma2.4}$(\romannumeral 5)$, we obtain $$\lambda _1\text{(}d_1,a_1,r-m\chi _{\left( x_{12},L \right]}\tilde{w}_2,(0,L))>\lambda _1\text{(}d_1,a_1,r-m\chi _{\left( x_{11},L \right]}\tilde{w}_1,(0,L)),$$ which implies that $\lambda _1\text{(}d_1,a_1,r-m\chi _{\left( x_1,L \right]}\tilde{w},(0,L))$ is strictly increasing with respect to $x_1$.
	Combining the continuity and strict monotonicity of $\lambda _1\text{(}d_1,a_1,r-m\chi _{\left( x_1,L \right]}\tilde{w},(0,L))$ with respect to $m$ and applying the implicit function theorem, we conclude that $m^*(x_1)$ is continuous and monotonically increasing with respect to $x_1$.
	
	Finally, to  study  the limits of $m^*(x_1)$, denote $\lambda _1\text{(}d_1,a_1,r-m\chi _{\left( x_1,L \right]}\tilde{w},
	\left( 0,L \right) \text{)}=:\lambda_1(x_1,m)$. Obviously $\lim\limits_{x_1\to0}m^*(x_1)=m^*(0)$. We conclude that  $m^*(0)>0$.  Otherwise, $m^*(0)=0$ satisfies $\lambda_1(0,m^*(0))=\lambda_1(d_1,a_1,r,(0,L))>0$, which contradicts with $\lambda_1(x_1,m^*(x_1))=0$.
	From \eqref{eq3.9},  Proposition~\ref{propB} and the definition of $m^*(x_1)$, we have $\lim\limits_{x_1\to L^*}\lim\limits_{m\to{+\infty}}\lambda_1(x_1,m)=\lambda _{1}^{D}\left( d_1,a_1,r,\left( 0,x_1 \right) \right)=0$ and $\lim\limits_{x_1\to L^*}\lambda_1(x_1,m^*(x_1))=0$. This, combined with Lemma~$\ref{lemma2.4}$ and the monotonic decreasing property of $\lambda_1(x_1,m)$ with respect to $m$, shows that $\lim\limits_{x_1\to L^*} m^*(x_1)=+\infty$.
\end{proof}

An analogous proof of Lemma~\ref{lemma3.2} gives the following results.
\begin{lemma}\label{lemma3.3}
Assume that  $0<a_1<\sqrt{4d_1r}$ and $L>L_1^*$. Then the following statements are true:

$(\romannumeral 1)$ If $x_1\ge L^*$ $(\text{large protection zone})$, then $(0,\tilde{w})$ is unstable;

$(\romannumeral 2)$ If $x_1<L^*$ $(\text{small protection zone})$, fixed all parameters except for $x_1$ and $q$, then there exists $q^*(x_1)>0$ satisfying $\lambda _1\left( d_1,a_1,r-m\chi _{\left( x_1,L \right]}\tilde{w},\left( 0,L \right) \right) =0$ such that the following statements hold:

$\quad$$(\romannumeral 2.1)$ If $q>q^*(x_1)$, then $(0,\tilde{w})$ is unstable;

$\quad$$(\romannumeral 2.2)$ If $0<q<q^*(x_1)$, then $(0,\tilde{w})$ is locally asymptotically  stable.

Furthermore, $q^*(x_1)$ is a continuous function of $x_1$ and strictly decreases on $(0,L)$, with $\lim\limits_{x_1\to 0}q^*(x_1)=q^*(0)>0$ and $\lim\limits_{x_1\to L^*} q^*(x_1)=0$.
\end{lemma}
Thus,  Theorem \ref{th1.1} can be derived from Lemmas~\ref{l1}, \ref{lemma3.2}, \ref{lemma3.3} and Proposition~\ref{propA}.
\begin{lemma}\label{lemma3.4}
 Assume that $0< a_1<\sqrt{4d_1r}$ and $L>L_1^*$, and that all parameters except for $h$ and $a_2$ are fixed.  Then we have the following result:

	$(\romannumeral 1)$ If $x_1\ge L^*$, then $(0,\tilde{w})$ is unstable;

	$(\romannumeral 2)$ If $x_1<L^*$, then there exists $h^{**}>0$ satisfying $\lambda _1(d_1,a_1,r-\frac{mh^{**}}{q}\chi _{\left( x_1,L \right]},\left( 0,L \right) )=0$ such that the following statements hold:

	$(\romannumeral 2.1)$ If $h<h^{**}$, then $(0,\tilde{w})$ is unstable;

	$(\romannumeral 2.2)$ If $h>h^{**}$, then there exists $a_2^*(h)>0$ satisfying $\lambda _1{(}d_1,a_1,r-m\chi _{\left( x_1,L \right]}\tilde{w},\left( 0,L \right) )=0$ such that $(0,\tilde{w})$ is linearly stable for $a_2<a_2^*(h)$, and unstable for $a_2>a_2^*(h)$.
	
	Moreover, $a_2^*(h)$ is continuous and  increases with respect to $h\in(h^{**},+\infty)$ with $\lim\limits_{h\to h^{**}}a_2^*(h)=0$ and $\lim\limits_{h\to{+\infty}}a_2^*(h)=+\infty$.
\end{lemma}
\begin{proof}
	From Lemmas \ref{lemma 2.2}, \ref{lemma2.3} and \ref{lemma2.4}, $\lambda _1\text{(}d_1,a_1,r-m\chi _{\left( x_1,L \right]}\tilde{w},\left( 0,L \right) \text{)}$ is continuous with respect to $a_2$ and $h$, respectively, and  is strictly increasing with respect to $a_2$ and decreasing with respect to $h$.
	
	From \cite[Lemma 2.2]{lou2015evolution}, we known that if $0< a_1<\sqrt{4d_1r}$ and $L>L_1^*$, then $\lambda _1\left( d_1,a_1,r,\left( 0,L \right) \right)>0$. From Lemma~\ref{lemma 2.2} $(\romannumeral 6)$ one can obtain that
	\begin{align}\label{eq3.10}
			&\lim\limits_{a_2\to0}\lambda _1\text{(}d_1,a_1,r-m\chi _{\left( x_1,L \right]}\tilde{w},\left( 0,L \right) \text{)}=\lambda _1\text{(}d_1,a_1,r-\frac{mh}{q}\chi _{\left( x_1,L \right]},\left( 0,L \right) \text{)},\\
		&\lim\limits_{a_2\to{+\infty}}\lambda _1\text{(}d_1,a_1,r-m\chi _{\left( x_1,L \right]}\tilde{w},\left( 0,L \right) \text{)}=\lambda _1\left( d_1,a_1,r,\left( 0,L \right) \right)>0.
	\end{align}

	 Next, to study the stability distribution of $(0,\tilde{w})$ in the $h-a_2$ plane, we examine the sign of $\lambda _1\text{(}d_1,a_1,r-\frac{mh}{q}\chi _{\left( x_1,L \right]},\left( 0,L \right) \text{)}.$
	
	 From Lemma \ref{lemma3.1},
	 \begin{equation}\label{3.11}
	 	\lim\limits_{h\to{+\infty}}\lambda _1\text{(}d_1,a_1,r-\frac{mh}{q}\chi _{\left( x_1,L \right]},\left( 0,L \right) \text{)}=\lambda _{1}^{D}\left( d_1,a_1,r,\left( 0,x_1 \right) \right).
	 \end{equation}
	 This, combined with Proposition~\ref{propB} and the monotonic decreasing property of $\lambda _1\text{(}d_1,a_1,r-\frac{mh}{q}\chi _{\left( x_1,L \right]},\left( 0,L \right) \text{)}$ with respect to $h$, implies that $\lambda _1\text{(}d_1,a_1,r-\frac{mh}{q}\chi _{\left( x_1,L \right]},\left( 0,L \right) \text{)}>0$ for any $h\in(0,+\infty)$ when $x_1\ge L^*$. Therefore, $\lambda _1\text{(}d_1,a_1,r-m\chi _{\left( x_1,L \right]}\tilde{w},\left( 0,L \right) \text{)}>0$ for any $a_2>0$  when $x_1\ge L^*$. Consequently, it follows that $(0,\tilde{w})$ is unstable when $x_1\ge L^*$.
	
	Note that
	\begin{equation}\label{3.12}
		\lim\limits_{h\to0}\lambda _1\text{(}d_1,a_1,r-\frac{mh}{q}\chi _{\left( x_1,L \right]},\left( 0,L \right) \text{)}=\lambda _1\left( d_1,a_1,r,\left( 0,L \right) \right)>0.
	\end{equation}
	Again from Proposition~\ref{propB},  $\lambda _{1}^{D}\left( d_1,a_1,r,\left( 0,x_1 \right) \right)<0$ when $x_1<L^*$. Then from \eqref{3.11} and \eqref{3.12}, there exists $h^{**}>0$ satisfying $\lambda _1\text{(}d_1,a_1,r-\frac{mh^{**}}{q}\chi _{\left( x_1,L \right]},\left( 0,L \right) \text{)}=0$, such that
\begin{equation}\label{hl}
	\lambda _1\text{(}d_1,a_1,r-\frac{mh}{q}\chi _{\left( x_1,L \right]}\tilde{w},\left( 0,L \right) \text{)}
		\begin{cases}
			<0,&for\quad h>h^{**},\\
            =0,&for\quad h=h^{**},\\
			>0,&for\quad 0<h<h^{**}.
		\end{cases}
	\end{equation}
This, combined with  \eqref{eq3.10}, implies that
	\begin{equation*}
		\lim\limits_{a_2\to0}\lambda _1\text{(}d_1,a_1,r-m\chi _{\left( x_1,L \right]}\tilde{w},\left( 0,L \right) \text{)}
		\begin{cases}
			<0,&for\quad h>h^{**},\\
			>0,&for\quad 0<h<h^{**}.
		\end{cases}
	\end{equation*}
By the monotonicity of $\lambda _1\text{(}d_1,a_1,r-m\chi _{\left( x_1,L \right]}\tilde{w},\left( 0,L \right) \text{)}$ with respect to $a_2$, we deduce that if $0<h<h^{**}$, then $\lambda _1\text{(}d_1,a_1,r-m\chi _{\left( x_1,L \right]}\tilde{w},\left( 0,L \right) \text{)}>0$ for all $a_2\leq 0$. So $(0,\tilde{w})$ is unstable if $0<h<h^{**}$. If $h>h^{**}$, then there exists $a_2^*(h)>0$ such that $\lambda _1\text{(}d_1,a_1,r-m\chi _{\left( x_1,L \right]}\tilde{w},\left( 0,L \right) \text{)}=0$. Consequently, $(0,\tilde{w})$ is linearly stable for $0\leq a_2<a_2^*(h)$ and unstable for $a_2>a_2^*(h)$.
	
	Combining the continuity and monotonicity of $\lambda _1\text{(}d_1,a_1,r-m\chi _{\left( x_1,L \right]}\tilde{w},\left( 0,L \right) \text{)}$ with respect to $a_2$ and $h$, and using the implicit function theorem, we concluded that $a_2^*(h)$ is continuous and monotonically increasing with respect to $h$.

 To  study  the limits of $a_2^*(h)$, denote $\lambda _1\text{(}d_1,a_1,r-m\chi _{\left( x_1,L \right]}\tilde{w},\left( 0,L \right) \text{)}=:\lambda_1(a_2,h)$. From \eqref{eq3.10}, \eqref{hl} and the definition of $h^{**}$, we have $\lim\limits_{h\to h^{**}}\lim\limits_{a_2\to 0}\lambda_1(a_2,h)=0$ and $\lim\limits_{h\to h^{**}}\lambda_1(a_2^*(h),h)=0$. Then by Lemma~\ref{eq2.4} and the monotonicity of $\lambda_1(a_2,h)$ with respect to $a_2$, we get that $\lim\limits_{h\to h^{**}}a_2^*(h)=0$.
	
	Suppose $\lim\limits_{h\to{+\infty}}a_2^*(h):=\beta\in (0,+\infty)$. Then there exists $\zeta_1>0$  such that $\lvert a_2^*(h)-\beta \rvert<1$ for $h>\zeta_1$. On the other hand, $\lim\limits_{h\to{+\infty}}\lambda_1(\beta+1,h)=\lambda _{1}^{D}\left( d_1,a_1,r,\left( 0,x_1 \right) \right):= M<0$. Then there exists $\zeta_2>0$ such that $\lambda_1(\beta+1,h)<\frac{3}{2}M<0$  for $h>\zeta_2$. Set $\zeta=max\{\zeta_1\text{,}\zeta_2\}$.  Combining $\lambda_1(a^*_2(\zeta+1),\zeta+1)=0$ and considering that $\lambda_1(a_2,h)$ is monotonically increasing with respect to $a_2$, we obtain $\lambda_1(\beta+1,\zeta+1))>0$, which leads to a contradiction. Hence $\beta=+\infty$. The proof is completed.
\end{proof}

Following an analogous proof of Lemma~\ref{lemma3.4}, we have the following results.
\begin{lemma}\label{lemma3.5}
 Assume that $0< a_1<\sqrt{4d_1r}$ and $L>L_1^*$, and that all parameters except for $q$ and $a_2$ are fixed. Then we have the following results:

	$(\romannumeral 1)$ If $x_1\ge L^*$, then $(0,\tilde{w})$ is unstable;

	$(\romannumeral 2)$ If $x_1<L^*$, then there exists $q^{**}>0$ satisfying $\lambda _1(d_1,a_1,r-\frac{mh}{q^{**}}\chi _{\left( x_1,L \right]},\left( 0,L \right)) =0$ such that the following statements hold:

	$(\romannumeral 2.1)$ If $q>q^{**}$, then $(0,\tilde{w})$ is unstable;

	$(\romannumeral 2.2)$ If $0<q<q^{**}$, then there exists $\bar{a}_2(q)>0$ satisfying $\lambda _1(d_1,a_1,r-m\chi _{\left( x_1,L \right]}\tilde{w},\left( 0,L \right)) =0$ such that $(0,\tilde{w})$ is stable for $a_2<\bar{a}_2(q)$ and unstable for $a_2>\bar{a}_2(q)$.
	
	Moreover, $\bar{a}_2(q)$ is  continuous and strictly decreases with respect to $q\in(0,q^{**})$ with $\lim\limits_{q\to q^{**}}\bar{a}_2(q)=0$ and $\lim\limits_{q\to0}\bar{a}_2(q)=+\infty$.
	\end{lemma}
\begin{lemma}\label{lemma3.6}
	Assume that $0< a_1<\sqrt{4d_1r}$ and $L>L_1^*$, and that all parameters except $m$ and $h$ are fixed. Then the following results are true:

	$(\romannumeral 1)$ If $x_1\ge L^*$, then $(0,\tilde{w})$ is unstable;

	$(\romannumeral 2)$ If $x_1<L^*$, then there exists $\hat{h}(m)>0$ satisfying $
	\lambda _1(d_1,a_1,r-m\chi _{\left( x_1,L \right]}\tilde{w},\left( 0,L \right)) =0$ such that $(0,\tilde{w})$ is stable for $h>\hat{h}(m)$ and unstable for $0<h<\hat{h}(m)$.
	
	Furthermore, $\hat{h}(m)$ is continuous and  strictly decreases with respect to $m\in(0,+\infty)$ with $\lim\limits_{m\to 0}\hat{h}(m)=+\infty$ and $\lim\limits_{m\to{+\infty} }\hat{h}(m)=0$.
\end{lemma}
\begin{proof}
	From  Lemma~\ref{lemma 2.2} and Lemma~\ref{lemma2.4},
	$	\lambda _1\text{(}d_1,a_1,r-m\chi _{\left( x_1,L \right]}\tilde{w},\left( 0,L \right) \text{)}	$ is strictly decreasing with respect to $h$ and
	\begin{align*}
		&\lim\limits_{h\to0}\lambda _1\text{(}d_1,a_1,r-m\chi _{\left( x_1,L \right]}\tilde{w},\left( 0,L \right) \text{)}=\lambda _1\left( d_1,a_1,r,\left( 0,L \right) \right),
		\\
		&\lim\limits_{h\to{+\infty}}\lambda _1\text{(}d_1,a_1,r-m\chi _{\left( x_1,L \right]}\tilde{w},\left( 0,L \right) \text{)}=\lambda _{1}^{D}\left( d_1,a_1,r,\left( 0,x_1 \right) \right).
      	\end{align*}
     The results of $(\romannumeral 1)$ and $(\romannumeral 2)$ can be proven in a similar manner to the proof of Lemma~\ref{lemma3.4}. Combining the continuity and monotonicity of $\lambda _1\text{(}d_1,a_1,r-m\chi _{\left( x_1,L \right]}\tilde{w},\left( 0,L \right) \text{)}$ with respect to $m$ and $h$, and using the implicit function theorem, we can  conclude that $\hat{h}(m)$ is continuous and strictly decreases with respect to $m$.

      Next, we show that $\lim\limits_{m\to 0}\hat{h}(m)=+\infty$ and $\lim\limits_{m\to{+\infty} }\hat{h}(m)=0$. For convenience, let $\lambda _1(d_1,a_1,r-m\chi _{\left( x_1,L \right]}\tilde{w},\left( 0,L \right))=:\lambda_1(m,h)$. Consequently,  $$0\le \tau_1:=\lim\limits_{m\to{+\infty} }\hat{h}(m) <\lim\limits_{m\to 0}\hat{h}(m)=:\tau_2\le+\infty.$$
     For $x_1<L^*$, combining \eqref{eq3.9} with Proposition~\ref{propB} yields
     \begin{equation}\label{e1}
     	\lim\limits_{m\to{+\infty}}\lambda_1(m,h)<0, \quad(\forall h>0).
     \end{equation}
     Based on the definition of $\hat{h}(m)$, we have
     \begin{equation}\label{e2}
     	\lim\limits_{m\to{+\infty}}\lambda_1(m,\hat{h}(m))=\lim\limits_{m\to{+\infty}}\lambda_1(m,\tau_1)=0.
     \end{equation}
     If $\tau_1>0$, \eqref{e1} and \eqref{e2} lead to a contradiction. Then we conclude that $\tau_1=0$.

     To show $\lim\limits_{m\to 0}\hat{h}(m)=+\infty$,  suppose to the contrary that $\hat{h}(m)$ is bounded above as $m\to 0$, then $\tau_2<+\infty$. Then there exists $\delta_1>0$, such that $|\hat{h}(m)-\tau_2|<\frac{1}{3}$ for all $m\in(0,\delta_1)$. From \eqref{eq2.12}, $\lim\limits_{m\to 0}\lambda_1(m,\tau_2+\frac{1}{3})=\lambda _1\left( d_1,a_1,r,\left( 0,L \right) \right) =: M_1>0$. Then there exists $\delta_2>0$, such that $\lambda_1(m,\tau_2+\frac{1}{3})>\frac{1}{2}M_1>0$ for all $m\in(0,\delta_2)$. Let $\delta=min\{\frac{1}{3}\delta_1\text{,}\frac{1}{3}\delta_2\}$. Then $\lambda_1(\delta,\hat{h}(\delta))=0$. Considering that $\lambda_1(m,h)$ is monotonically decreasing with respect to $h$, it follow that $\lambda_1(\delta,\tau_2+\frac{1}{3})<0$, which lead to a contradiction. The proof is completed.
\end{proof}
With the help of Lemmas~\ref{l1}, \ref{lemma3.4}, \ref{lemma3.5} and \ref{lemma3.6}, we have Theorem~\ref{th1.2}.
\begin{proof}[Proof of Theorem~$\ref{th1.3}$]
	By Lemma~\ref{lemma2.4}, we obtain that $\lambda _1\text{(}d_1,a_1,r-m\chi _{\left( x_1,L \right]}\tilde{w},\left( 0,L \right) \text{)}$ is strictly decreasing with respect to $a_1>0$ and $\lim\limits_{a_1\to{+\infty}}\lambda _1\text{(}d_1,a_1,r-m\chi _{\left( x_1,L \right]}\tilde{w},\left( 0,L \right) \text{)}=-\infty$. Subsequently, we only need to examine the sign of $\lambda _1\text{(}d_1,0,r-m\chi _{\left( x_1,L \right]}\tilde{w},\left( 0,L \right) \text{)}$.
	Again using Lemma~\ref{eq2.4}, one obtains that $\lambda _1\text{(}d_1,0,r-m\chi _{\left( x_1,L \right]}\tilde{w},\left( 0,L \right) \text{)}
	$ is strictly decreasing with respect to $d_1$. Meanwhile, we have
	$$\lim\limits_{d_1\to 0}\lambda _1\text{(}d_1,0,r-m\chi _{\left( x_1,L \right]}\tilde{w},\left( 0,L \right) \text{)}=r>0,$$ and $$
	\lim\limits_{d_1\to{+\infty}}\lambda _1\text{(}d_1,0,r-m\chi _{\left( x_1,L \right]}\tilde{w},\left( 0,L \right) \text{)}=r-\frac{m\int_{x_1}^{L}\tilde{w}dx}{L}.$$
	From the above facts, one can deduce that parts $(\romannumeral 1)$ and $(\romannumeral 2)$) hold, and $d_1^*(r)$ is increasing  with respect to $r$.
	
	Next, we show that $\lim\limits_{r\to{0}}d_1^*(r)=0$ and $\lim\limits_{r\to{r^*}}d_1^*(r)=+\infty$.
	
	 Suppose that $d_1^*(r)$ is bounded below as $r\to 0^+$. By the monotonicity of $d_1^*(r)$ with respect to $r$, there exists a constant $K_1>0$ such that $\lim\limits_{r\to{0}}d_1^*(r)=K_1$. From Lemma~\ref{lemma2.4}, we have
	\begin{align*}
	0=&\lim\limits_{r\to{0}}\lambda _1\text{(}d_1^*(r),0,r-m\chi _{\left( x_1,L \right]}\tilde{w},\left( 0,L \right) \text{)}=\lambda _1\text{(}K_1,0,-m\chi _{\left( x_1,L \right]}\tilde{w},\left( 0,L \right) \text{)}\\
	&<\lim\limits_{d_1\to{0}}\lambda _1\text{(}d_1,0,-m\chi _{\left( x_1,L \right]}\tilde{w},\left( 0,L \right) \text{)}=\max\limits_{x\in[0,L]}(-m\chi _{\left( x_1,L \right]}\tilde{w})=0,
	\end{align*}
	which is a contradiction. Hence, $\lim\limits_{r\to{0}}d_1^*(r)=0$.
	
	 Suppose that $d_1^*(r)$ is bounded above as $r\to r^*$. By the monotonicity of $d_1^*(r)$ with respect to $r$, there exists a constant $K_2>0$ such that $\lim\limits_{r\to{r^*}}d_1^*(r)=K_2<+\infty$. From Lemma~\ref{lemma2.4}, we have
	\begin{align*}
		0=&\lim\limits_{r\to{r^*}}\lambda _1\text{(}d_1^*(r),0,r-m\chi _{\left( x_1,L \right]}\tilde{w},\left( 0,L \right) \text{)}=\lambda _1\text{(}K_2,0,r^*-m\chi _{\left( x_1,L \right]}\tilde{w},\left( 0,L \right) \text{)}\\
		&>\lim\limits_{d_1\to{+\infty}}\lambda _1\text{(}d_1,0,r^*-m\chi _{\left( x_1,L \right]}\tilde{w},\left( 0,L \right) \text{)}=\frac{\int_{0}^{L}r^*-m\chi _{\left( x_1,L \right]}\tilde{w}dx}{L}=0,
	\end{align*}
	which is a contraction. Thus, we have $\lim\limits_{r\to{r^*}}d_1^*(r)=+\infty$. The proof is finished.
\end{proof}

Next, we investigate stability of the coexistence steady state  (if it exists) of system (\ref{eq1.1}).
Linearizing \eqref{eq1.1} at the coexistence steady state $(u,w)$, we obtain
 \begin{equation}\label{eq3.1}
 	\begin{cases}
 		d_1\phi_{xx}-a_1\phi_x+(r-2u-m\chi _{\left( x_1,L \right]}w)\phi-m\chi _{\left( x_1,L \right]}u\psi=\lambda\phi,&x\in (0,L),\\
 		d_2\psi_{xx}-a_2\psi_x-pw\phi-(q+pu)\psi=\lambda\psi,&x\in (x_1,L),\\
 		d_1\phi_x(0)-a_1\phi(0)=\phi_x(L)=0,\\
 		d_2\psi_x(x_1)-a_2\psi(x_1)=\psi_x(L)=0.
 	\end{cases}
 \end{equation}
  From Krein-Rutman Theorem \cite{krein1948linear}, problem (\ref{eq3.1}) admits a principal eigenvalue, denoted by $\lambda_1^*$, and its corresponding principal eigenfunction $(\phi_1,\psi_1)$ can be chosen such that $\phi_1>0$ in $[0,L]$ and $\psi_1<0$ in $[x_1,L]$. Furthermore,
 \begin{equation*}
 	if \quad \lambda_1^*
 	\begin{cases}
 		<0,\\
 		=0,\\
 		>0,
 	\end{cases}
 	\quad
 	\text{then }(u,w) \text{ is}
 	\begin{cases}
 		\text{linearly stable},\\
 		\text{neutrally stable},\\
 		\text{linearly unstable}.
 	\end{cases}
 \end{equation*}

\begin{lemma}\label{lemma3.7}
	Assume that $0< a_1<\sqrt{4d_1r}$, $L>L_1^*$,  $h\ge\tilde{w}^2(L)$ and $pme^{|\frac{a_1}{d_1}-\frac{a_2}{d_2}|L}\le 1$. If system $(\ref{eq1.1})$ admits a coexistence steady state $(u,w)$, then the principal eigenvalue of \eqref{eq3.1} satisfies $\lambda_1^*<0$, i.e., $(u,w)$ is linearly stable.
\end{lemma}

\begin{proof}
	If system $(\ref{eq1.1})$ admits a coexistence steady state $(u,w)$, then $(u,w)$ satisfies \eqref{3.13}.
	Recall that $\lambda_1^*$  and $(\phi_1,\psi_1)$ are the principal eigenvalue and principal eigenfunction  associated with $(u,w)$, where $(\lambda_1^*,\phi_1,\psi_1)$ satisfy \eqref{eq3.1}, and $\phi_1>0$ on $[0,L]$ and $\psi_1>0$ on $[x_1,L]$.
	
	 Multiplying the first equation of (\ref{eq3.1}) by $u$ and the first equation of (\ref{3.13}) by $\phi_1$, and subtracting the resulting equations, one gets
	\begin{equation}\label{eq3.12}
		\lambda_1^*\phi_1u=(d_1\phi_{1,xx}-a_1\phi_{1,x})u-(d_1u_{xx}-a_1u_x)\phi-u^2(\phi_1+m\chi _{\left( x_1,L \right]}\psi_1).
	\end{equation}
	Similarly, one can derive from the second equations in (\ref{eq3.1}) and (\ref{3.13}) that
	\begin{equation}\label{eq3.13}
		\lambda_1^*\psi_1w=(d_2\psi_{1,xx}-a_2\psi_{1,x})w-(d_2w_{xx}-a_2w_x)\psi_1-(pw^2\phi_1+h\psi_1).
	\end{equation}
	Then multiplying (\ref{eq3.12}) with $\frac{\phi_1^2}{u^2}e^{-\frac{a_1}{d_1}x}$ and integrating the resulting equation over $(0,L)$, one get
	\begin{equation}
	\begin{aligned}\label{eq3.14}
		\lambda_1^*\int_{0}^{L}e^{-\frac{a_1}{d_1}x}\frac{\phi_1^3}{u}dx=&\int_{0}^{L}e^{-\frac{a_1}{d_1}x}\frac{\phi_1^2}{u^2}((d_1\phi_{1,xx}-a_1\phi_{1,x})u-(d_1u_{xx}-a_1u_x)\phi_1)dx\\
		&-\int_{0}^{L}e^{-\frac{a_1}{d_1}x}\phi_1^2(\phi_1+m\chi _{\left( x_1,L \right]}\psi_1)dx\\
		=:&I_1-I_2.
	\end{aligned}
	\end{equation}
Similarly to \cite{huang2022}, we have
$$I_1=-\int_{0}^{L}2d_1e^{-\frac{a_1}{d_1}x}\frac{\phi_1^3}{u}(\frac{u_x}{u}-\frac{\phi_{1,x}}{\phi_1})^2dx\le 0.$$
In a similar way, multipling (\ref{eq3.13})  with $\frac{\psi_1^2}{w^2}e^{-\frac{a_2}{d_2}x}$ and integrating the resulting equation over $(x_1,L)$, we get
\begin{equation}\label{eq3.15}
\begin{aligned} \lambda_1^*\int_{x_1}^{L}e^{-\frac{a_2}{d_2}x}\frac{\psi_1^3}{w}dx=&\int_{x_1}^{L}e^{-\frac{a_2}{d_2}x}\frac{\psi_1^2}{w^2}((d_2\psi_{1,xx}-a_2\psi_{1,x})w-(d_2w_{xx}-a_2w_x)\psi_1)dx\\
	&-\int_{x_1}^{L}e^{-\frac{a_2}{d_2}x}\frac{\psi_1^2}{w^2}(pw^2\phi_1+h\psi_1)dx\\
	=:&J_1-J_2.
\end{aligned}
\end{equation}
 Here, $$J_1=-\int_{x_1}^{L}2d_2e^{-\frac{a_2}{d_2}x}\frac{\psi_1^3}{w}(\frac{w_x}{w}-\frac{\psi_{1,x}}{\psi_1})^2dx\ge 0.$$

Next, we consider two cases to complete the proof.

 \textbf{cases 1:} $\frac{a_1}{d_1}\ge \frac{a_2}{d_2}$. Multiplying (\ref{eq3.15}) by $m^3$ and subtracting the resulting equation from (\ref{eq3.14}), one arrives at
 \begin{equation}\label{eq3.16}
 	\begin{aligned}
 		\lambda_1^* \large[\int_{0}^{L} e^{-\frac{a_1}{d_1}x} \frac{\phi_1^3}{u} dx
 		- \int_{x_1}^{L}e^{-\frac{a_2}{d_2}x}\frac{(m\psi_1)^3}{w}dx\large] &= I_1 - I_2 - m^3 J_1 + m^3J_2.\\
 		&\le m^3J_2-I_2,
 	\end{aligned}
 \end{equation}
with
\begin{equation}\label{eq3.170}
\begin{aligned}
	m^3J_2-I_2=&\int_{x_1}^{L}e^{-\frac{a_2}{d_2}x}\frac{\psi_1^2}{w^2}(m^3pw^2\phi_1+m^3h\psi_1)dx\\
	&-\int_{0}^{L}e^{-\frac{a_1}{d_1}x}\phi_1^2(\phi_1+m\chi _{\left( x_1,L \right]}\psi_1)dx\\	<&\int_{x_1}^{L}e^{-\frac{a_1}{d_1}x}\frac{h}{w^2}(m\psi_1)^3dx+\int_{x_1}^{L}e^{-\frac{a_1}{d_1}x}\phi_1(m\psi_1)^2[mpe^{\frac{a_1}{d_1}x-\frac{a_2}{d_2}x}]dx\\
	&-\int_{x_1}^{L}\phi_1^3e^{-\frac{a_1}{d_1}x}dx-\int_{x_1}^{L}\phi_1^2(m\psi_1)e^{-\frac{a_1}{d_1}x}dx\\
\end{aligned}
\end{equation}
Assume that $mpe^{\frac{a_1}{d_1}L-\frac{a_2}{d_2}L}\le1$ and $\underset{x\in \left[ x_1,L \right]}{\min}\frac{h}{w^2}\ge1$. With the facts $\phi_1>0$ on $[0,L]$, $\psi_1<0$ on $[x_1,L]$, we have
\begin{equation}\label{eq3.17}
\begin{aligned}
	m^3J_2-I_2\le&\int_{x_1}^{L}e^{-\frac{a_1}{d_1}x}(m\psi_1)^3dx+\int_{x_1}^{L}e^{-\frac{a_1}{d_1}x}\phi_1(m\psi_1)^2dx\\
	&-\int_{x_1}^{L}\phi_1^3e^{-\frac{a_1}{d_1}x}dx-\int_{x_1}^{L}\phi_1^2(m\psi_1)e^{-\frac{a_1}{d_1}x}dx\\
	=&\int_{x_1}^{L}e^{-\frac{a_1}{d_1}x}(m\psi_1-\phi_1)(m\psi_1+\phi_1)^2dx\\
	\le&0.
\end{aligned}
\end{equation}
Combining \eqref{eq3.16}, \eqref{eq3.170} and \eqref{eq3.17}, we conclude that $\lambda_1^*<0$.

\textbf{cases 2:} $\frac{a_1}{d_1}< \frac{a_2}{d_2}$. In this case, multiplying (\ref{eq3.14}) by $p^3$ and subtracting to (\ref{eq3.15}), one obtains
\begin{equation}\label{eq3.18}
	 \begin{aligned}
		\lambda_1^* [\int_{0}^{L} e^{-\frac{a_1}{d_1}x} \frac{(p\phi_1)^3}{u} dx
		- \int_{x_1}^{L}e^{-\frac{a_2}{d_2}x}\frac{\psi_1^3}{w}dx] &= p^3I_1 - p^3I_2 - J_1+J_2\\
		&\le J_2-p^3I_2.
	\end{aligned}
\end{equation}

  Similarly to \eqref{eq3.17}, one can show that
  \begin{equation}\label{eq3.190}
  	\begin{aligned}
  		J_2-p^3I_2=&\int_{x_1}^{L}e^{-\frac{a_2}{d_2}x}\frac{\psi_1^2}{w^2}(pw^2\phi_1+h\psi_1)dx\\
  		&-\int_{0}^{L}e^{-\frac{a_1}{d_1}x}\phi_1^2(p^3\phi_1+p^3m\chi _{\left( x_1,L \right]}\psi_1)dx\\
  		<&\int_{x_1}^{L}e^{-\frac{a_2}{d_2}}p\phi_1\psi_1^2dx+\int_{x_1}^{L}e^{-\frac{a_2}{d_2}}\frac{h}{w^2}\psi_1^3dx-\int_{x_1}^{L}e^{-\frac{a_2}{d_2}}(p\phi_1)^3dx\\
  		&-\int_{x_1}^{L}e^{-\frac{a_2}{d_2}}\psi_1(p\phi_1)^2[mpe^{\frac{a_2}{d_2}x-\frac{a_1}{d_1}x}]dx\\
  	\end{aligned}
\end{equation}
Assume that $mpe^{\frac{a_1}{d_1}L-\frac{a_2}{d_2}L}\le1$ and $\underset{x\in \left[ x_1,L \right]}{\min}\frac{h}{w^2}\ge1$. With the facts $\phi_1>0$ on $[0,L]$, $\psi_1<0$ on $[x_1,L]$, we have
\begin{equation}\label{eq3.19}
\begin{aligned}
	J_2-p^3I_2	\le&\int_{x_1}^{L}e^{-\frac{a_2}{d_2}}p\phi_1\psi_1^2dx+\int_{x_1}^{L}e^{-\frac{a_2}{d_2}}\psi_1^3dx-\int_{x_1}^{L}e^{-\frac{a_2}{d_2}}(p\phi_1)^3dx\\
  		&-\int_{x_1}^{L}e^{-\frac{a_2}{d_2}}\psi_1(p\phi_1)^2\\
  		=&\int_{x_1}^{L}e^{-\frac{a_2}{d_2}}(\psi_1-p\phi_1)(\psi_1+p\phi_1)^2dx\\
  		\le&0.
  	\end{aligned}
  \end{equation}

  Combining \eqref{eq3.18}, \eqref{eq3.190} and \eqref{eq3.19}, we have $\lambda_1^*<0$.

 From Lemma~\ref{l1}$(\romannumeral 2)$, we have $\underaccent{\sim}{w}<w(x)<\tilde{w}(x)$ in $[x_1,L]$. Combining with Lemma~\ref{lemma 2.2}$(\romannumeral 1)$, we conclude that if $h\ge \tilde{w}^2(L)$, then $\underset{x\in \left[ x_1,L \right]}{\min}\frac{h}{w^2}\ge1$.
Then the  above analysis shows $\lambda_1^*<0$ provided that $mpe^{|\frac{a_1}{d_1}L-\frac{a_2}{d_2}L|}\le1$ and $h\ge \tilde{w}^2(L)$, which implies that the coexistence steady state $(u,w)$ is linearly stable.
\end{proof}
Thus, Theorem~\ref{th1.4} can be derived from Lemma~\ref{lemma3.7}, Proposition~\ref{propA} and Theorem~\ref{th1.1}.
\subsection{Dynamics of system (\ref{eq1.2})}
In this subsection, we analyze the dynamics of system~\eqref{eq1.2}.
It is well-known that the linear stability of $(0,\bar{w})$ is determined by the following eigenvalue
problem
\begin{equation*}
	\begin{cases}
		d_1\phi_{xx}-a_1\phi_x+(r-m\chi _{\left[ 0,x_2 \right)}\bar{w})\phi=\lambda\phi,&x\in (0,L),\\
		d_1\phi_x(0)-a_1\phi(0)=\phi_x(L)=0.
	\end{cases}
\end{equation*}
 Similarly to \cite[Corollary 2.10]{sweers1992strong}, we obtain that
\begin{equation*}
	if \quad \lambda_1(d_1,a_1,r-m\chi _{\left[ 0,x_2 \right)}\bar{w},(0,L))
	\begin{cases}
		<0,\\
		=0,\\
		>0,
	\end{cases}
	\quad
	\text{then }(0,\bar{w}) \text{ is}
	\begin{cases}
		\text{linearly stable},\\
		\text{neutrally stable},\\
		\text{linearly unstable}.
	\end{cases}
\end{equation*}

Let $\lambda_{1}^{N}(d, a, r, (x_2, L))$ represent the principal eigenvalue of the following characteristic equation
	\begin{equation}\label{eq4.3}
		\begin{cases}
			d\varphi_{xx}-a\varphi_x+r\varphi=\lambda^N\varphi,&x_2<x<L,\\
			\varphi(x_2)=\varphi_x(L)=0.
		\end{cases}
	\end{equation}
Then we have the following results.
\begin{proposition}\label{propC}
	Assume that $a<\sqrt{4dr}$.
	\begin{equation}\label{eq4.2}
	If\quad	L-x_2
		\begin{cases}
			< L^*, \\
			= L^* ,\\
			> L^*,
		\end{cases}
		\quad
		\text{then }\quad \lambda_{1}^{N}(d, a, r, (x_2, L))
		\begin{cases}
			< 0, \\
			= 0, \\
			> 0.
		\end{cases}
	\end{equation}	
\end{proposition}

 For system~\eqref{eq1.2}, Theorem \ref{th1.5} and \ref{th1.6} can be proved by similar method  to those in Lemma~\ref{lemma3.6} and Theorem~\ref{th1.3}.

Linearizing
system~\eqref{eq1.2} at the coexistence steady state $(u,w)$, we obtain
\begin{equation}\label{eq4.1}
	\begin{cases}
		d_1\phi_{xx}-a_1\phi_x+(r-2u-m\chi _{\left[ 0,x_2 \right)}w)\phi-m\chi _{\left[ 0,x_2 \right)}u\psi=\lambda\phi,&x\in (0,L),\\
		d_2\psi_{xx}-a_2\psi_x-pw\phi-(q+pu)\psi=\lambda\psi,&x\in (0,x_2),\\
		d_1\phi_x(0)-a_1\phi(0)=\phi_x(L)=0,\\
		d_2\psi_x(0)-a_2\psi(0)=d_2\psi_x(x_2)-a_2\psi(x_2)=0.
	\end{cases}
\end{equation}
 From Krein-Rutman Theorem \cite{krein1948linear},  problem~\eqref{eq4.1} admits a principal eigenvalue, denoted by $\bar{\lambda}_1$, with the corresponding principal eigenfunction $(\bar{\phi}_1,\bar{\psi}_1)$, which can be chosen such that $\bar{\phi}_1>0$ in $[0,L]$ and $\bar{\psi}_1<0$ in $[0,x_2]$. Moreover,
 \begin{equation*}
 	If \quad \bar{\lambda}_1
 	\begin{cases}
 		<0,\\
 		=0,\\
 		>0,
 	\end{cases}
 	\quad
 	\text{then }(u,w) \text{ is}
 	\begin{cases}
 		\text{linearly stable},\\
 		\text{neutrally stable},\\
 		\text{linearly unstable}.
 	\end{cases}
 \end{equation*}
Similar to Lemma~\ref{lemma3.7}, we have the following results.
\begin{lemma}\label{lemma4.1}
Assume that $0< a_1<\sqrt{4d_1r}$, $L>L_1^*$,  $h\ge\bar{w}^2(x_2)$ and $pme^{|\frac{a_1}{d_1}-\frac{a_2}{d_2}|x_2}\le 1$. If system $(\ref{eq1.2})$ admits a coexistence steady state $(u,w)$,  then the principal eigenvalue of \eqref{eq4.1} satisfies  $\bar{\lambda}_1<0$, i.e.,  $(u,w)$ is linearly stable.
\end{lemma}
 Theorem \ref{th1.7} follows from Lemma~\ref{lemma4.1}, Theorem~\ref{th1.5}, and Proposition~\ref{propA}.

\section{Numerical simulations}
In this section, in order to investigate the combined effects of toxicant advection and protected zone size
interactions on population persistence, we will numerically analyzes the stability regions of the coexistence steady state  in the $a_2\text{-}x_1$ plane for system~(\ref{eq1.1}) and  in the $a_2\text{-}z$ ($z:=L-x_2$) plane for system~(\ref{eq1.2}).

\begin{figure}[h]
	\centering
	\subfigure[]{\includegraphics[width=5cm,height=5cm]{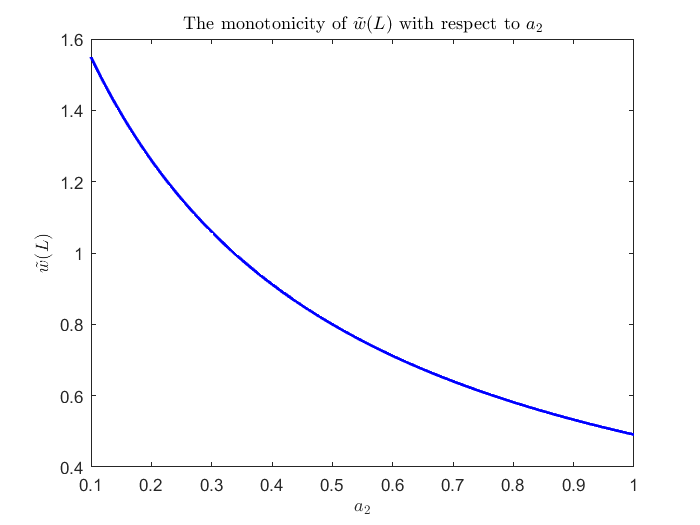}} \hfill
	\subfigure[]{\includegraphics[width=5cm,height=5cm]{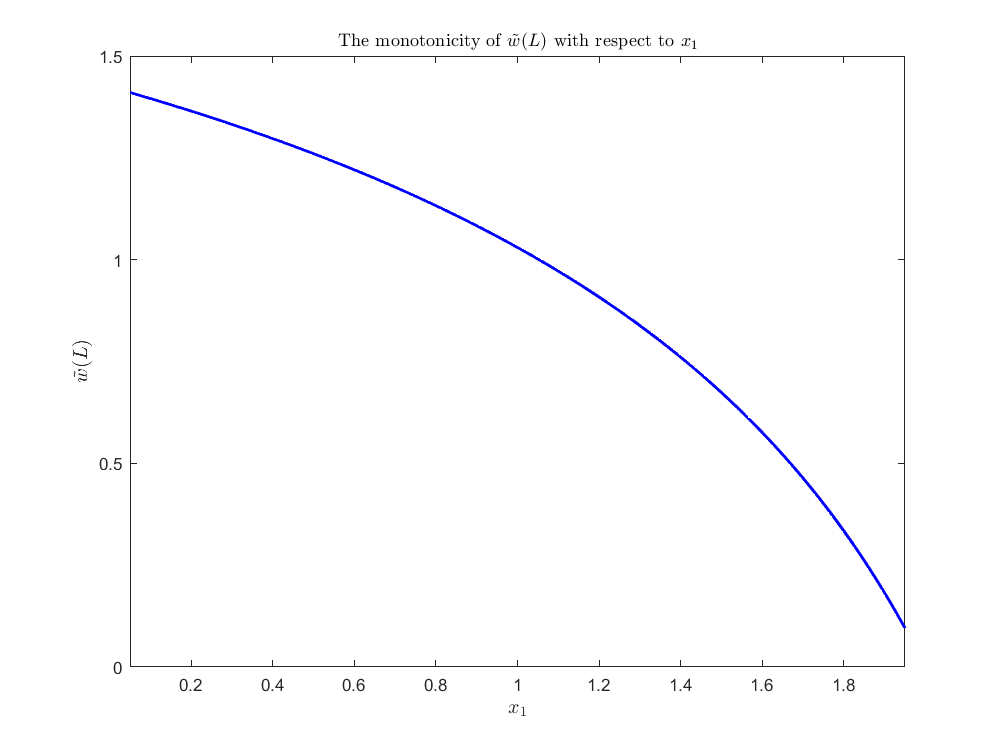}}
\subfigure[]{\includegraphics[width=5cm,height=5cm]{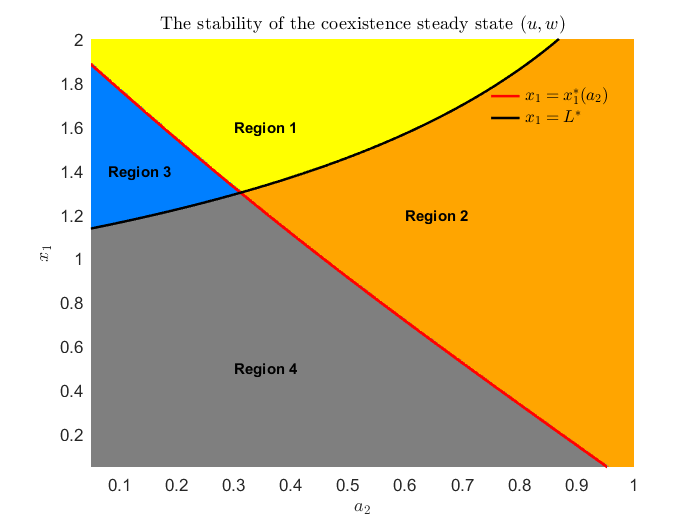}} \hfill
	\caption{ (a): $\tilde{w}(L)$ decreases monotonically with respect to $a_2$ for $x_1=0.5$. (b): $\tilde{w}(L)$ decreases monotonically with respect to $x_1$ for $a_2=0.2$. (c): Under the condition $\frac{a_1}{d_1}=\frac{a_2}{d_2}$, the stability of the coexistence steady state \((u, w)\) is divided into four regions in the \(a_2\text{-}x_1\) parameter plane. Other parameters: $d_1=d_2=0.5, L=2, h=0.4, q=0.2, r=1$.}
	\label{fig:6}
\end{figure}
Theorem~\ref{th1.4} provides a sufficient condition related to $\tilde{w}(L)$ to characterize the regions of population persistence in the $x_1\text{-}a_2$ parameter plane via the stability analysis of the coexistence steady state of system~\eqref{eq1.1}. 
Fig. \ref{fig:6}(a) and Fig. \ref{fig:6}(b) demonstrate the monotonicity of $\tilde{w}(L)$ with respect to $a_2$ and $x_1$, respectively. For the function $h=\tilde{w}^2(L)$, fixing all parameters except $a_2$ and $x_1$, and using the implicit function theorem together with $\tilde{w}(L)$ being monotonically decreasing in $x_1$, we conclude $x_1=x_1^*(a_2)$ is monotonic, which is consistent with the trend shown by the red solid line in Fig. \ref{fig:6}(c).

 Fig. \ref{fig:6}(c) shows that, in the $a_2\text{-}x_1$ parameter plane, under the setting where all parameters are fixed except $a_2$ and $x_1$, and $a_1<\sqrt{4d_1r}$ and $L>L_1^*$, the stability of the coexistence steady state \((u, w)\) is divided into four regions; selecting an appropriate value of $m$ in each region can ensure that \((u, w)\) is stable. The black solid line indicates that, with parameters \(d_1\) and \(r\) fixed, the critical value \(L^{\ast}\) of the protected area increases as the toxin convection velocity \(a_2\) increases.
  When \( \frac{a_1}{d_1} =\frac{a_2}{d_2} \), the stability conclusions for each region in Fig. \ref{fig:6}(c) are as follows:
 \begin{itemize}
 	\item Region 1 (\( h \ge \tilde{w}^2(L), x_1 \ge L^* \)): If \( m < p^{-1} \), the coexistence steady state \( (u, w) \) is globally asymptotically stable.
 	\item Region 2 (\( h \ge \tilde{w}^2(L), x_1 < L^* \)): If \( m < \min\{p^{-1}, m^*\} \), the coexistence steady state \( (u, w) \) is globally asymptotically stable.
 	\item Region 3 (\( h < \tilde{w}^2(L), x_1 \ge L^* \)): The coexistence steady state \( (u, w) \) is linearly stable.
 	\item Region 4 (\( h < \tilde{w}^2(L), x_1 < L^* \)): If \( m < m^* \), the coexistence steady state \( (u, w) \) is linearly stable.
 \end{itemize}

\begin{figure}[h]
	\centering
	\subfigure[]{\includegraphics[width=5cm,height=5cm]{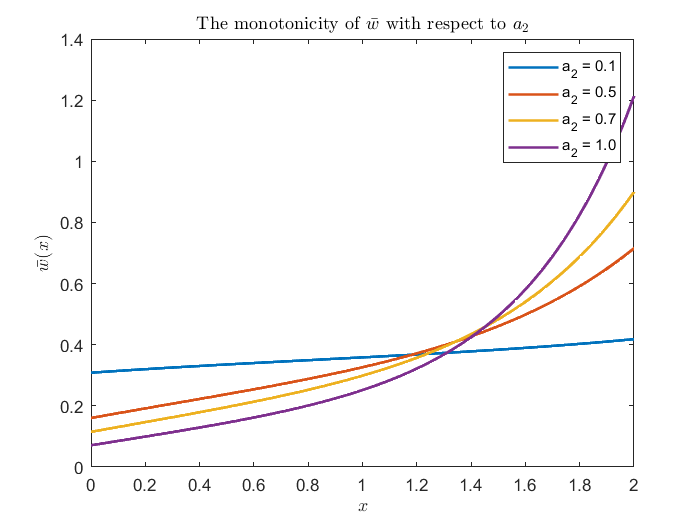}} \hfill
	\subfigure[]{\includegraphics[width=5cm,height=5cm]{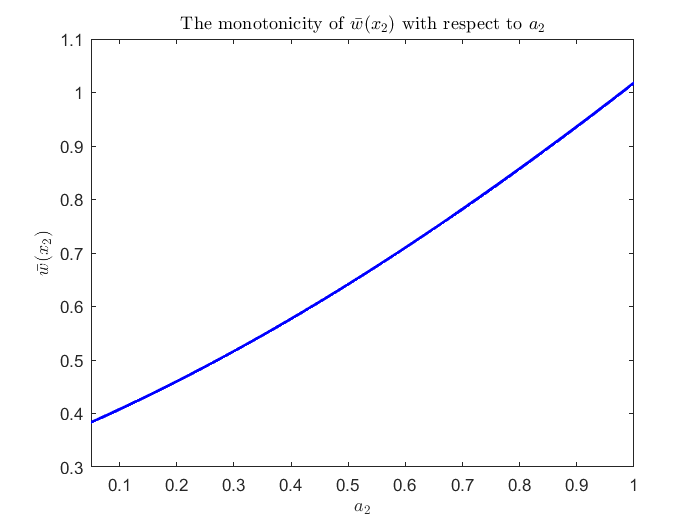}}
	\subfigure[]{\includegraphics[width=5cm,height=5cm]{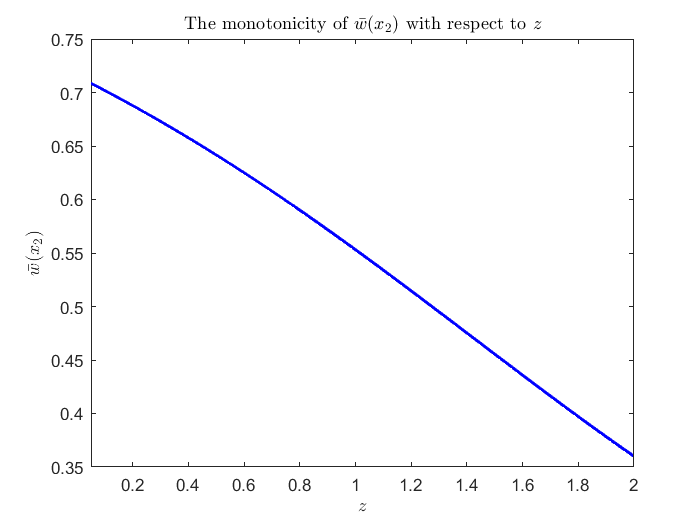}} \hfill
	\caption{(a): The dependence of $\bar{w}$ on the advection rate $a_2$ when only toxicant is present. (b): The monotonicity of $\bar{w}(x_2)$ with respect to $a_2$ for $z=0.5$. (c): The monotonicity of $\bar{w}(x_2)$ with respect to $z$ for $a_2=0.5$. Other parameters: $d_1=d_2=0.5, L=2, h=0.18, q=0.5$.}
	\label{fig:7}
\end{figure}

 For system \eqref{eq1.2}
 Fig. \ref{fig:7}(a) shows that $\bar{w}$ lacks monotonicity with respect to $a_2$ for $x\in(0,x_2)$. Consequently, it is impossible to obtain conclusions similar to Theorem~\ref{th1.2} through theoretical analysis—specifically, analyzing the stability region of $(0,\bar{w})$ in the parameter plane involving $a_2$. 
 while Fig. \ref{fig:7}(b) and Fig. \ref{fig:7}(c) demonstrate the monotonicity of $\bar{w}(x_2)$ with respect to $a_2$ and $z$, respectively. Thus, we can analyze the stability of the coexistence steady state \((u, w)\) of system~\eqref{eq1.2} based on the sufficient condition related to \(\bar{w}(x_2)\) provided by Theorem~\ref{th1.7}.
 For the function $h=\bar{w}^2(x_2)$, fixing all parameters except $a_2$ and $z$, and using the implicit function theorem together with $\bar{w}(x_2)$ being monotonically decreasing in $z$, we conclude $z=z^*(a_2)$ is monotonic,
 which is consistent with the trend shown by the red solid line in Fig. \ref{fig:8}.
 
  Fig. \ref{fig:8} shows that, in the $a_2\text{-}z$ parameter plane, under the setting where all parameters are fixed except $a_2$ and $z$, and $a_1<\sqrt{4d_1r}$ and $L>L_1^*$, the stability of the coexistence steady state \((u, w)\) is divided into four regions; selecting an appropriate value of $m$ in each region can ensure that \((u, w)\) is stable. The black solid line indicates that, with parameters \(d_1\) and \(r\) fixed, the critical value \(L^{\ast}\) of the protected area increases as the toxin convection velocity \(a_2\) increases.
 When \( \frac{a_1}{d_1} =\frac{a_2}{d_2} \), the stability conclusions for each region in Fig. \ref{fig:8} are as follows:
 \begin{itemize}
 	\item Region 1 (\( h \ge \bar{w}^2(L), z \ge L^* \)): If \( m < p^{-1} \), the coexistence steady state \( (u, w) \) is globally asymptotically stable.
 	\item Region 2 (\( h \ge \bar{w}^2(L), z < L^* \)): If \( m < \min\{p^{-1}, \bar{m}\} \), the coexistence steady state \( (u, w) \) is globally asymptotically stable.
 	\item Region 3 (\( h < \bar{w}^2(L), z \ge L^* \)): The coexistence steady state \( (u, w) \) is linearly stable.
 	\item Region 4 (\( h < \bar{w}^2(L), z < L^* \)): If \( m < \bar{m} \), the coexistence steady state \( (u, w) \) is linearly stable.
 \end{itemize}
\begin{figure}[h]
	\centering
	\includegraphics[width=6cm,height=5cm]{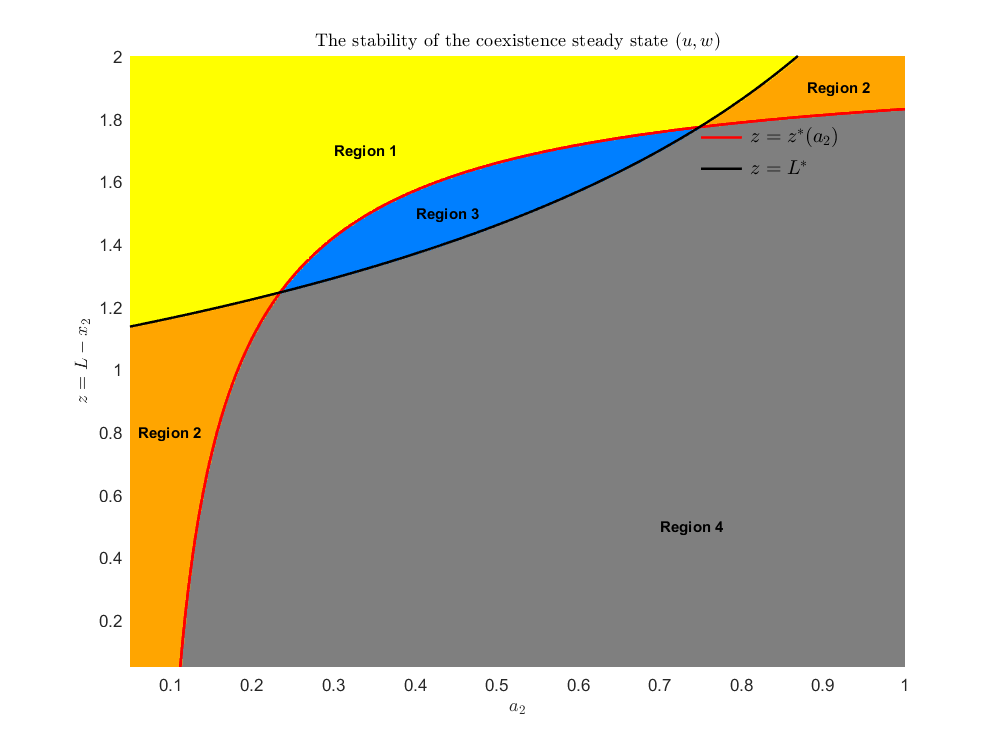}
	\caption{Under the condition $\frac{a_1}{d_1}=\frac{a_2}{d_2}$, the stability of the coexistence steady state \((u, w)\) is divided into four regions in the \(a_2\text{-}z\) parameter plane. Parameter values: $d_1=d_2=0.5, L=2, h=0.18, q=0.5, r=1$.}
	\label{fig:8}
\end{figure}
\section{Conclusion}

To protect aquatic species in polluted rivers, this paper develops a reaction-diffusion-advection model for population-toxicant interactions. The model incorporates Danckwerts boundary conditions and two types of boundary protection zones. For both upstream
and downstream protection zone configurations, analysis of toxicant-only steady states yields conditions for population persistence. Furthermore,  applying monotone
dynamical system theory and eigenvalue analysis,  we derive a sufficient condition for the global asymptotic stability of the coexistence steady state.


Specifically, this paper focuses on the comprehensive effects of the size of the protection zone ($x_1$/$L-x_2$), the population diffusion rate ($d_1$), the population/toxicant advection rate ($a_1$/$a_2$), the intrinsic growth rate of the population ($r$), the toxicant input rate ($h$), the toxicant effect coefficient
on population growth ($m$), and the contaminant discharge rate per unit ($q$) on the persistence of the population.

Our study shows that the population face extinction due to its large advection rate ($a_1\ge\sqrt{4d_1r}$) or a small  advection rate while a short habitat length ($0< a_1<\sqrt{4d_1r}$, $L<L^*_1$) no matter the protected zone is situated at the upstream or the downstream of the river.  When the advection rate of the population is small  and the  river
 length is large $(0< a_1<\sqrt{4d_1r}, L>L_1^*)$, the persistence of the population is determined by the combined interaction of the protection zone length and other biological parameters. Specifically,
\begin{itemize}
  \item When the protected area exceeds a critical spatial threshold ($x_1\ge L^*$, resp. $L-x_2\ge L^*$), population persistence occurs irrespective of toxicant exposure levels.
  \item Below this threshold ($x_1< L^*$, resp. $L-x_2< L^*$), persistence becomes dependent on additional factors. Specifically, if any of the following occurs: the toxicant effect coefficient on population growth ($m$) is sufficiently high, the contaminant discharge rate per unit ($q$) is sufficiently low, or the toxicant input rate ($h$) is sufficiently high, then a larger critical protection zone size is required to ensure persistence.
\end{itemize}

For a fixed protection  zone,
	\begin{itemize}
  \item When the toxicant input $h$ or the effect of toxicant on population growth $m$ is low, or the contaminant export rate $q$ is high,  population $u$ persists regardless of toxicant advection. Conversely, high $h$ or $m$, or low $q$, induces habitat-dependent survival: the population persists at high $a_2$ but goes extinct at low  $a_2$.
 	\item  For sufficiently small  $m, h>0$, the population $u$ can always persist. The critical contamination threshold  for the persistence of population $u$ decreases monotonically with increasing $m$.
  	\item  For fixed $m>0$,  the population $u$ can always persist when $q$ is sufficiently large. The critical detoxification threshold increases monotonically with  $m$.
 \item   If the intrinsic growth rate $r$ is sufficiently large, then a low advection rate $a_1$ promotes population persistence; If the intrinsic growth rate $r$ is small, then both a low diffusion rate $d_1$ and a low advection rate $a_1$
are required to ensure population persistence.
 \end{itemize}

Some of our future work will focus on extending the mathematical modeling and analysis to address the following issues:
(1) The impact of the protection zone on population dynamics under spatially heterogeneous exogenous toxicant input rates into the river;
(2) The influence of the protection zone on population dynamics when considering alternative boundary conditions.
A deeper investigation into these questions will provide enhanced scientific and theoretical support for water resource and environmental management authorities. This will facilitate improved monitoring and management of species diversity within aquatic ecosystems, ultimately contributing to the optimization of ecological conservation strategies.

\bibliographystyle{plain}
\bibliography{references}
\end{document}